\def\ARXIVVERSION{1}
\documentclass[sigconf]{acmart} %% CCS: DO NOT REMOVE
\newif\ifarxiv
\ifdefined\ARXIVVERSION
  \arxivtrue
\else
  \arxivfalse
\fi

\usepackage{filecontents}

\usepackage{amsmath}
\usepackage{booktabs,tabularx,caption}
\usepackage{listings}
\usepackage{listings-rust}%
\usepackage{xcolor}
\usepackage{algorithmicx}%
\usepackage{algorithm}
\usepackage[noend]{algpseudocode}%
\usepackage{pifont}
\usepackage{enumitem}
\usepackage{colortbl} % For coloring table cells
\usepackage{tikz}[node]%
\usepackage{wrapfig}
\usepackage{tcolorbox}%
\usepackage{colortbl}
\usepackage[table]{xcolor}
\usepackage[nomessages]{fp}% http://ctan.org/pkg/fp
\usepackage{xfp}%
\usepackage{placeins}%
\usepackage{xspace}
\usepackage{adjustbox}
\usepackage{ulem}%
\usepackage{amsthm}
\usepackage{hyperref}
\usepackage{ulem}
\usepackage{stmaryrd}
\usepackage{mathrsfs}
\usepackage{multirow}

\usepackage{mathpartir}

\usepackage{enumitem}
\usepackage{bbm}
\usepackage{footnote}

\newcommand{\sys}{\textsc{ZEBRA}\xspace}

\newcommand{\nzkvmreal}{5\xspace}

\newcommand{\nfixed}{3\xspace}
\newcommand{\nconfirmed}{6\xspace}
\newcommand{\nbugsrealworld}{11\xspace}

\newcommand{\isreal}{\mathfrak{s}}
\newcommand{\opa}{\mathfrak{a}}
\newcommand{\opb}{\mathfrak{b}}
\newcommand{\opc}{\mathfrak{c}}
\newcommand{\programcounter}{\mathfrak{pc}}
\newcommand{\nextprogramcounter}{\mathfrak{pc'}}

\newtheorem{theorem}{Theorem}[section]
\newtheorem{lemma}[theorem]{Lemma}

\newtheorem{definition}[theorem]{Definition}

\newcommand{\rangeverification}{input range verification\xspace}

\colorlet{critcol}{red!22}  
\colorlet{highcol}{orange!28}
\colorlet{medcol}{yellow!28}
\colorlet{infocol}{cyan!18} 
\colorlet{obscol}{black!11}
\newcommand{\sevbadge}[2]{\colorbox{#1}{\small\bfseries\strut\kern3pt#2\kern3pt}}
\newcommand{\Crit}{\sevbadge{critcol}{Critical}}
\newcommand{\High}{\sevbadge{highcol}{High}}
\newcommand{\Medium}{\sevbadge{medcol}{Medium}}
\newcommand{\Obs}{\sevbadge{obscol}{Design Flaw}}
\newcommand{\Info}{\sevbadge{infocol}{Design Choice}}

\lstdefinelanguage{Lean}{
  morekeywords={
    theorem, lemma, def, structure, abbrev, namespace, end, import,
    fun, let, in, if, then, else, match, with, do, return,
    Prop, Type, Sort, by, intro, exact, rw, simp, apply, refine,
    constructor, cases, induction, rfl, where
  },
  morekeywords=[2]{
    Config, EventEncoding, EventSet, Trace, Row, Finset, Function,
    TableEncodesEvents, TableGeneratorFaithful, ALU
  },
  sensitive=true,
  morecomment=[l]{--},
  morecomment=[s]{/-}{-/},
  morestring=[b]",
  literate=
    {↔}{{$\leftrightarrow$}}1
    {→}{{$\to$}}1
    {∀}{{$\forall$}}1
    {∃}{{$\exists$}}1
    {∈}{{$\in$}}1
    {∧}{{$\wedge$}}1
    {₁}{{$_1$}}1
    {₂}{{$_2$}}1
}

\newcommand{\hideaki}[1]{\textcolor{black}{#1}}

\AtBeginDocument{%
  }

\setcopyright{acmlicensed}
\copyrightyear{2026}
\acmYear{2026}
\setcopyright{cc}
\setcctype{by}
\acmConference[CCS '26]{Proceedings of the 2026 ACM SIGSAC Conference on Computer and Communications Security}{November 15--19, 2026}{The Hague, Netherlands}
\acmBooktitle{Proceedings of the 2026 ACM SIGSAC Conference on Computer and Communications Security (CCS '26), November 15--19, 2026, The Hague, Netherlands}
\acmDOI{10.1145/3830454.3846593}
\acmISBN{979-8-4007-2871-6/2026/11}

\begin{document}

%% allowing the author to define a "short title" to be used in page
%% headers.

\title{Efficient Branch-and-Bound Testing and Verification of zkVMs}%% CCS: you MUST provide a title


%% CCS: at submission time, the submission MUST be anonymized. Hence
%% authors MUST be commented out.

% \affiliation{%

\author{Hideaki Takahashi}
\affiliation{%
    \institution{Columbia University}
    \city{New York}
    \country{USA}}
\email{ht2673@columbia.edu}

\author{Suman Jana}
\affiliation{%
    \institution{Columbia University}
    \city{New York}
    \country{USA}}
\email{suman@cs.columbia.edu}

\author{Junfeng Yang}
\affiliation{%
    \institution{Columbia University}
    \city{New York}
    \country{USA}}
\email{junfeng@cs.columbia.edu}

% \affiliation{%

% \affiliation{%

% \affiliation{%

% \affiliation{%

% \affiliation{%

% \affiliation{%


\begin{abstract} %% CCS: an abstract MUST be provided.

Zero-knowledge virtual machines (zkVMs) enable verifiable execution of general-purpose programs by translating virtual machine semantics into algebraic constraints over execution traces. The correctness of these constraints is critical: a single missing or incorrect constraint can admit forged proofs (under-constrained) or reject valid executions (over-constrained). Existing approaches do not provide meaningful guarantees at production scale: fuzzers and unit tests often miss bugs, SMT solvers struggle with the size and non-linearity of constraints, and theorem provers require substantial manual effort.

We present \sys, a fully automated verification and bug-detection framework grounded in a precise correctness criterion: for a given program and input, the constraint system must admit exactly one valid execution trace—no more (soundness) and no fewer (completeness). This reduces zkVM verification to a solution-set cardinality problem over a canonical trace space, where redundancies such as null-row padding and non-deterministic permutations are eliminated prior to counting.

To compute cardinality tractably, \sys lifts analysis from finite-field witnesses to an integer interval lattice, exploiting a structural sparsity property of zkVM constraints: across \nzkvmreal real-world zkVMs, constraints utilize only 14.0\% of their theoretical connectivity capacity on average. This sparsity enables tight interval propagation with limited over-approximation error. \sys performs a parallel branch-and-bound search that, unlike fuzzing, either produces a concrete counterexample or certifies the absence of violations within a bounded input region.

We evaluate \sys on five real-world zkVMs. \sys discovers \nbugsrealworld{} previously unknown vulnerabilities, including critical flaws enabling control-flow hijacking and proof forgery—none detected by a state-of-the-art fuzzer; \nconfirmed{} have already been independently confirmed and \nfixed{} have been fixed by developers. Compared to SMT-based verification, \sys is $51.5\times$ faster, verifies $16.5$ percentage point more instances, and its range verification provides up to $63\times$ efficiency gain over repeated single-input verification.
\end{abstract}

%mention that fuzzers cannot even execute point-wise verification

%input region

%of which \nconfirmed were confirmed by developers and \nfixed fixed



%\ccsdesc[500]{Do Not Use This Code~Generate the %Correct Terms for Your Paper}

 %% CCS: DO NOT REMOVE but you MAY update


\maketitle

%% CCS: You MAY change the title and,
                       %% obviously, add text, sections, figures,
                       %% tables, etc. 
%% CCS: For help and more latex examples, refer to
%% `sample-sigconf.tex', provided in the distribution
%% https://portalparts.acm.org/hippo/latex_templates/acmart-primary.zip 
%%


%% CCS: to preserve anonymity, NO acknowledgements to fundings, projects or persons should be used at
%% submission time
%% CCS: this section MAY be used to acknowledge the use of AI when used only for minor editorial improvements (e.g., grammar, spelling, or light style polishing) 
% \begin{acks}
% This paper was edited for grammar using [Tool Name].
% \end{acks}

\section{Introduction}
Zero-knowledge proofs (ZKPs) enable verifiable computation by reducing program execution to the satisfaction of algebraic constraints over a finite field. This process, arithmetization, has traditionally been performed manually using domain-specific languages such as Circom~\cite{belles2022circom} and Noir~\cite{aztec2024noir}, a workflow that is both labor-intensive and error-prone: even minor mismatches between intended semantics and generated constraints can produce severe soundness or completeness violations~\cite{wen2024practical,takahashi2025zkfuzz,pailoor2023automated}.

\textbf{zkVMs \& the shifting trust boundary.} Zero-Knowledge Virtual Machines (zkVMs) address the above-mentioned issue by fixing a universal arithmetization at the level of a virtual machine, eliminating the need for per-program circuit construction~\cite{dokchitser2023zero}. Developers write programs in standard languages (e.g., Rust, C++, Go), which are compiled to a conventional ISA, typically RISC-V. Program execution produces a trace of machine states; the zkVM verifies that this trace satisfies a fixed constraint system encoding valid state transitions. Arithmetization is thus performed only once for the VM semantics and can be reused across arbitrary programs. zkVMs have gained rapid adoption in blockchain, machine learning, and supply-chain verification~\cite{hassanzadeh2025constraint,benno2026jolt,wang2026zkagent,ron2026verifiable}.

 This shift in abstraction fundamentally changes the trust boundary: correctness now depends entirely on the zkVM's arithmetization, the translation of virtual machine logic (CPU state transitions, memory consistency, ALU operations) into a polynomial constraint system over a finite field~\cite{ke2025consistency,kwan2024verifying}. Concretely, for a fixed program and input, this system must admit exactly the valid execution traces, no more (soundness) and no fewer (completeness). In practice, these constraints operate over large two-dimensional trace tables with thousands of cells and non-linear constraints, expressed via Algebraic Intermediate Representation (AIR) and lookup arguments~\cite{dokchitser2023zero}. Even a single missing or incorrect constraint can introduce critical vulnerabilities: under-constrained systems admit forged proofs, while over-constrained systems reject valid executions. Real-world zkVMs have exhibited both failure modes; for instance, an under-constrained bug in RISC Zero was severe enough to warrant a \$50{,}000 bounty~\cite{RiscZero2025SecurityDisclosure}.

\textbf{Limitations of existing approaches.} Existing approaches to securing zkVM constraints suffer from fundamental limitations. Dynamic methods such as unit testing and fuzzing uncover shallow bugs by sampling concrete inputs~\cite{hochrainer2025arguzz}, but offer no completeness guarantee: subtle constraint violations hidden in the vast state space of a large prime field often remain undetected, as shown in our evaluation. Formal verification via {SMT solvers is precise but brittle}~\cite{ke2025consistency}. Interactive theorem provers such as Lean 4 can validate constraints with high assurance~\cite{kwan2024verifying}, but require substantial manual effort and specialized expertise. Therefore, no existing method simultaneously achieves scalability, automation, and high semantic coverage over the full constraint system, {while ZEBRA trades precision for scalable region reasoning with a fully automated pipeline.}

\noindent\textbf{Our approach.}
We reduce zkVM verification to a solution-set cardinality estimation problem and solve it using interval-based branch-and-bound over a canonicalized trace space. For a fixed program~$P$ and input~$x$, let $\mathcal{C}(P,x,T)$ denote the constraint system evaluated on trace~$T$. We define the \textit{solution set}
\[
  \mathcal{S}(P,x) \;=\; \bigl\{\,T \;\big|\; \mathcal{C}(P,x,T)\,\bigr\}.
\]
We find that correctness of zkVM constraints reduces to a cardinality condition on~$\mathcal{S}$: the system is \emph{complete} if $|\mathcal{S}|\ge 1$, \textit{sound} if
$|\mathcal{S}|\le 1$, and fully correct iff $|\mathcal{S}|=1$. Under-constrained bugs manifest as $|\mathcal{S}|>1$; over-constrained bugs as $|\mathcal{S}|=0$. This formulation provides a unified, precise criterion for zkVM correctness.

We classify the cardinality of $\mathcal{S}(P,x)$ via the following three components. Rather than computing $|\mathcal{S}|$ exactly, we classify regions of the search space as empty, singleton, or ambiguous, enabling early certification without exhaustive enumeration.

\smallskip
\noindent\textbf{Canonicalized table space.}
We eliminate representational redundancy in execution traces by introducing a canonicalization layer that maps each trace table to a unique canonical representation. Without canonicalization, padding, row permutations, or unconstrained witness cells can cause the same execution to be counted multiple times, rendering solution-set cardinality meaningless.

The central correctness requirement is injectivity: distinct canonical representations must correspond to distinct executions, so that $|\mathcal{S}|$ faithfully reflects the number of distinct executions. A key structural observation is that real-world zkVM designs admit canonicalizers of a simple and uniform form—in practice requiring on average fewer than $130$ lines of code per zkVM. Across five independent zkVM implementations, we formally verify the injectivity of all canonicalizers in Lean~4,\footnote{We verify these properties at the algorithmic specification level; verifying full implementation correctness in Lean is orthogonal and left to future work.} establishing that canonicalization is well-defined over equivalence classes of traces and preserves solution-set cardinality.

\smallskip
\noindent\textbf{Sound interval propagation.}
We evaluate constraints over interval ranges instead of exact finite-field values, soundly over-approximating all feasible executions. Modular arithmetic creates the main challenge: wrap-around at the field characteristic breaks standard interval behavior, so we design conservative propagation rules that preserve soundness. During branch-and-bound search, if propagation eliminates all interval assignments, then $\mathcal{S}(P,x) = \emptyset$; if all remaining assignments collapse to a singleton region, then $|\mathcal{S}(P,x)| = 1$.

Our interval propagation benefits from the sparsity of zkVM constraint systems, which designers maintain to reduce verification cost. Because each constraint touches only a small subset of variables, updates propagate locally and introduce limited over-approximation error. Consistent with this, real-world zkVMs exhibit low connectivity: their density $\rho = T_{\mathrm{arith}}(\mathcal{C}) / (|\mathcal{C}| \times m)$—measuring the fraction of variables touched per constraint relative to a dense system—averages $0.140$ across different backends (Appendix~\ref{appendix:sp}). Branch-and-bound refinement then reduces the remaining error by partitioning the trace space, enabling aggressive pruning while preserving soundness.

\smallskip
\noindent\textbf{Parallel branch-and-bound certification.}
We use interval propagation to guide a branch-and-bound search over the trace lattice, certifying solution-set cardinality. Each branch splits the trace space into independent subproblems, enabling natural parallelism in contrast to SMT solvers. The same refinement also tightens interval bounds, reduces over-approximation error, and accelerates convergence. This design supports localized verification: by restricting inputs to a bounded region, \sys{} can certify correctness within that region without requiring a global proof over all inputs. If the search terminates with $|\mathcal{S}| = 1$ it certifies local correctness; otherwise, it returns concrete counterexamples. \hideaki{A fixed-point iteration over local certificates formally characterizes how they compose toward global correctness, analogous to composing layerwise bound propagation in neural network verification~\cite{alphabetacrown} albeit at a greater computational cost. We leave the implementation of this iteration as future work focusing on local verification in this paper.}

\smallskip
\noindent\textbf{Evaluation.}
We implement \sys{} in Rust and evaluate it on \nzkvmreal{} real-world zkVM backends: Pico, SP1, Sphinx, Valida, and Ziren. \sys{} discovers \nbugsrealworld{} previously unknown zero-day vulnerabilities, of which \nconfirmed{} were confirmed by developers and \nfixed{} already fixed, demonstrating that the approach scales to production constraint systems and surfaces bugs that fuzzing and formal methods have missed. Our code is publicly available at
https://github.com/Koukyosyumei/ZEBRA. \ifarxiv\else The appendices referenced throughout this paper are also available in
the extended version on arXiv.\footnote{\url{ARXIV_URL}}
\fi

\smallskip
\noindent In summary, this paper makes the following contributions:
\begin{itemize}[leftmargin=1.5em,itemsep=2pt]
  \item A \textbf{canonicalization-based correctness model} that reduces zkVM verification to solution-set cardinality over a canonical table space. We show that real-world zkVMs admit compact canonicalizers with formally verified correctness properties in Lean, ensuring that solution-set cardinality reflects the number of distinct executions rather than representational artifacts.
  
  \item \textbf{Sound interval propagation} over polynomial constraints in a modular field, enabling over-approximate reasoning about feasible executions with formal soundness guarantees.
  
  \item A \textbf{parallel branch-and-bound algorithm} for cardinality classification that either uncovers vulnerabilities or certifies local correctness by proving that the solution set is empty or singleton within a bounded input region, rather than merely testing for the existence of satisfying traces.
  
  \item \textbf{Empirical evaluation} identifying \nbugsrealworld{} zero-day vulnerabilities in widely deployed zkVM projects.
\end{itemize}

\section{Background}

This section reviews the foundations of Zero-Knowledge Proof (\S~\ref{subsec:zkp}) and Zero-Knowledge Virtual Machine (\S~\ref{subsec:zkvm}).

\subsection{Zero Knowledge Proof (ZKP)}
\label{subsec:zkp}

\begin{figure*}[!th]
    \centering
    \includegraphics[width=1.0\linewidth]{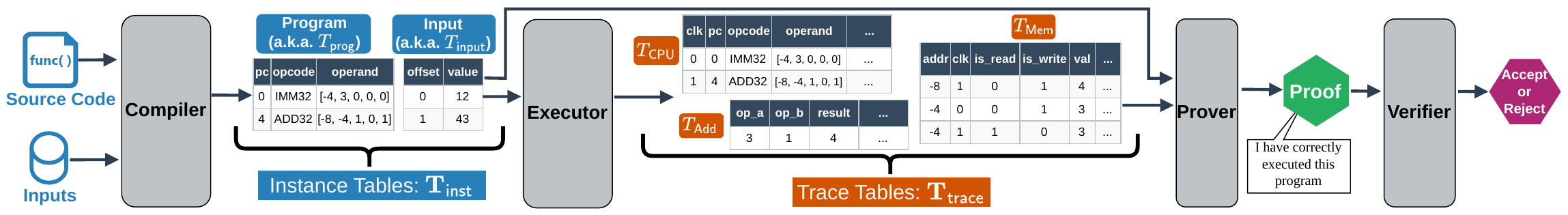}
    \caption{Typical workflow of a zkVM. The compiler converts source code into instance tables ($\mathbf{T}_{\mathrm{inst}}$), which the executor expands into trace tables ($\mathbf{T}_{\mathrm{trace}}$) representing the full execution history. If $\mathbf{T}_{\mathrm{trace}}$ satisfies all algebraic constraints with respect to $\mathbf{T}_{\mathrm{inst}}$, the prover can generate a final proof to convince the verifier that the target program was executed correctly.}
    \label{fig:zkvm-overview}
\end{figure*}

%shorten the first paragraph.

A \textit{Zero-Knowledge Proof (ZKP)} is a protocol where a prover convinces a verifier that a statement $R(x,w)$ is true without revealing anything beyond its validity~\cite{lavin2024survey}. It satisfies completeness, soundness, and zero knowledge: completeness means true statements are accepted, soundness means false statements are rejected except with negligible probability, and zero knowledge means the proof reveals nothing about $w$.

\smallskip
\noindent \textbf{ZK Programming.} A ZKP system requires users to represent the target statement as polynomial constraints $C$, which is complex and error-prone, as high-level logic does not map directly to polynomial constraints. ZK DSLs such as Circom~\cite{belles2022circom} and Noir~\cite{aztec2024noir} address this by letting developers specify both witness computation P and constraints C, which are then compiled into executable code and constraint systems. Despite this, maintaining consistency between computation and constraints is difficult due to complex DSLs and finite-field arithmetic, leading to bugs that cause under-constrained or over-constrained circuits and resulting vulnerabilities.

\subsection{Zero Knowledge Virtual Machine (zkVM)}
\label{subsec:zkvm}

While traditional ZK programming systems require developers to learn new DSLs and define arithmetic circuits for specific tasks, a \textit{Zero-Knowledge Virtual Machine} (zkVM) provides a more general-purpose abstraction by enabling the proof of execution for any program written in a standard high-level language (e.g., Rust, C++, or Go). A zkVM typically adopts a modular architecture consisting of a guest program, a prover, and a verifier (see Fig.~\ref{fig:zkvm-overview}). After the execution of a guest program, a prover of a zkVM produces a cryptographic proof attesting that the guest program was correctly executed, and the verifier can check its correctness succinctly.

\smallskip
\noindent \textbf{Trace Tables.} A zkVM targets a specific ISA (e.g., RISC-V or custom ones like Cairo~\cite{goldberg2021cairo} and Valida~\cite{thomas2025valida}) and executes compiled machine code to produce execution traces of state transitions. These traces are encoded as finite-field matrices called \textit{trace tables}, where each row represents an execution event, and each column represents a machine state component, such as the clock. Static data like binaries and public inputs are also treated as read-only \textit{instance tables}.

\smallskip
\noindent \textbf{Arithmetization of the VM.} A zkVM converts execution logic into polynomial constraints over trace tables to prove correctness without revealing traces. Modern zkVMs use two types of constraints: \textit{Algebraic Intermediate Representation} (AIR) constraints~\cite{ben2018scalable,hassanzadeh2025constraint} enforce local transitions between rows, and \textit{lookup constraints} ensure global consistency across tables by matching entries between them~\cite{habock2022multivariate, thomas2025valida, openvm2025whitepaper}. The prover generates a proof that all constraints hold, which the verifier can succinctly check.

\smallskip
\noindent \textbf{Example of Arithmetization.}  Consider a simplified execution trace for a program consisting of three instructions: \texttt{STORE -4 3}, which writes the immediate value 3 to memory address -4; \texttt{STORE -8 1}, which writes the value 1 to address -8; and \texttt{ADD -12 -4 -8}, which reads from addresses -4 and -8, then writes the sum 4 to address -12. The table column design is based on Valida~\cite{thomas2025valida}, whereas many columns are omitted. %(see Fig.~\ref{fig:examples-constraints} in Appendix~\ref{appendix:sp}).

%Fig.~\ref{fig:examples-constraints} illustrates

\textit{1. Program Table ($T_{\mathsf{prog}}$)} is a read-only instance table that defines the compiled machine code. Each row specifies the instruction (\texttt{opcode}) and its associated destination (\texttt{dst}) and source (\texttt{src1}, \texttt{src2}) operands. A lookup constraint between the CPU Table $T_{\mathsf{CPU}}$ ensures that every instruction executed in ($T_{\mathsf{CPU}}$), defined by the tuple (\texttt{pc}, \texttt{opcode}, \texttt{dst}, \texttt{src1}, \texttt{src2}), actually exists within $T_{\mathsf{prog}}$.

\textit{2. CPU Table ($T_{\mathsf{CPU}}$)} serves as the main dynamic execution trace, where each row captures the current program counter (\texttt{pc}), the instruction being executed, and the values (\texttt{val1}, \texttt{val2}, \texttt{res}) processed during that cycle. Local AIR constraints ensure basic CPU logic, such as the validity of control flow and clock transitions, requiring that $\texttt{next.pc} = \texttt{curr.pc} + 4$ and $\texttt{next.clk} = \texttt{curr.clk} + 1$.

\textit{3. ADD Table ($T_{\mathsf{Add}}$)} validates arithmetic correctness. When the CPU executes an ADD instruction, a lookup constraint ensures that the tuple (\texttt{val1}, \texttt{val2}, \texttt{res}) exists in $T_{Add}$ as (\texttt{a}, \texttt{b}, \texttt{c}), where the AIR constraint $\texttt{a} + \texttt{b} = \texttt{c}$ holds.

%\smallskip
%\noindent \textbf{Comparison with DSL-based Approach.} Because the VM logic is static, the constraints remain the same regardless of the guest program being executed; only the execution trace changes based on the program's logic and inputs. This design brings a critical advantage to zkVMs. Developers can leverage existing libraries written in popular high-level languages and standard toolchains without needing to learn new DSLs, significantly reducing the complexity of manual circuit construction.

%a recently discovered under-constrained bug in RISC Zero, one of the most widely adopted RISC-V–based zero-knowledge virtual machines (zkVMs), was deemed severe enough to warrant a \$50K bounty payout for a single issue~\cite{RiscZero2025SecurityDisclosure}.

\section{Threat Model \& Motivating Example} 
{\bf Threat Model.} This work focuses on soundness and completeness violations in zkVMs arising from incorrect or incomplete constraint specifications.
Our adversary exploits flaws in the constraint system to either
(i) produce invalid execution traces that nevertheless satisfy the constraints (under-constrained behavior), or
(ii) cause valid executions \hideaki{generated by the VM executors} to be rejected by the constraints
(over-constrained behavior).

We explicitly exclude infrastructure-level issues, including bugs in
the compiler, executor, prover, or verifier implementations, as well
as attacks on the cryptographic protocol itself. We also exclude
zero-knowledge leakage, as confidentiality guarantees are
application-specific and orthogonal to the correctness of the constraint satisfaction~\cite{takahashi2025zkfuzz}.

\begin{figure}[!th]
    \centering
    \includegraphics[width=1.0\linewidth]{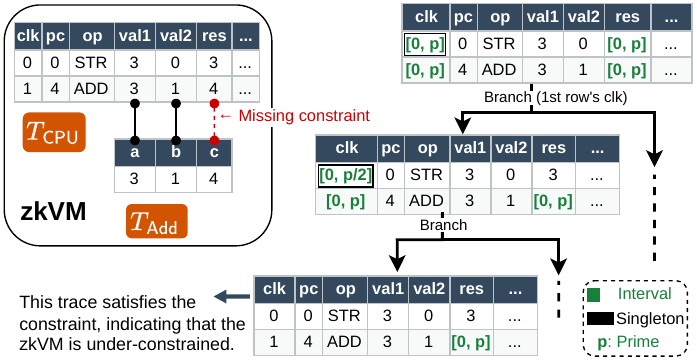}
    \caption{Motivating Example. Missing constraint allows multiple satisfying tables, leading to under-constrained bugs that \sys can identify via branch-and-bound. Green values indicate intervals.}
    \label{fig:motivating-example}
\end{figure}

\noindent
{\bf Motivating Example.} Before presenting the technical details, we illustrate a simplified yet representative bug to build intuition for how our method identifies constraint bugs. 

Consider a simple program with two instructions:
\texttt{STORE} and \texttt{ADD}, where the inputs to \texttt{ADD} are concretely given: $3$ and $1$
(see Fig.~\ref{fig:motivating-example}). In this case, the constraint linking the ADD result to the CPU state contains a
zkVM constraint bug: while the system enforces that
$(\texttt{val}_1, \texttt{val}_2, \texttt{res})$ appears in the ADD table,
it does not enforce that $\texttt{res}$ is uniquely determined by
$\texttt{val}_1$ and $\texttt{val}_2$. As a result, in addition to the intended execution trace
$\{\texttt{val}_1: 3, \texttt{val}_2: 1, \texttt{res}: 4\}$,
a second trace
$\{\texttt{val}_1: 3, \texttt{val}_2: 1, \texttt{res}: 7\}$
also satisfies the constraints, indicating an under-constrained violation.

\sys reasons about sets of executions by lifting table cells to intervals. For the ADD step, we initialize $\texttt{res} \in [0, p-1]$. In a correct system, constraint propagation enforces
$\texttt{res} = \texttt{val}_1 + \texttt{val}_2 = 4$,
collapsing the interval to the singleton $[4,4]$.
In the buggy system, there are multiple values (e.g., $4$ and $7$) satisfy the constraints,
enabling \sys to detect the under-constrained bug by producing
multiple concrete traces satisfying the constraints.

While this motivating example assumes that the program and input are concrete values, \S~\ref{subsec:rv} shows \sys can handle the case when they are given as range intervals.

%interval contains contains an under-constrained 

%\sys focuses on soundness and completeness vulnerabilities due to wrongly implemented constraints. The following types of vulnerabilities are out of scope: (1) vulnerabilities in the underlying compilation, proving, and verification systems (e.g., the Frozen Heart vulnerability caused by an insecure implementation of the Fiat-Shamir transformation~\cite{tang2024zero}); (2) bugs solely in the VM emulator logic, such as the discrepancied between the implementaton of VM emulator and the specification of target ISA; and (3) bugs solely in ZK protocol design that leak information and violate the zero-knowledge property~\cite{chaliasos2024sok}, inline with all prior detectors~\cite{circomspect,pailoor2023automated,jiang2025conscs,chen2024ac4,wen2024practical}---what constitutes a secret is highly application-dependent: e.g., in Validity Rollups~\cite{chen2022review}, ZK proofs are used to accelerate transactions rather than to conceal them.

% --------------
% background is too long. should be shorter than one page.

% Figure's 2 can be simplified. The caption might be longer

%Fig.~1 : make characters bigger

\section{Definitions of zkVM Vulnerabilities}

This section introduces a formal model of constraint systems and bugs in zkVMs. Our model abstracts away prover and cryptographic details, focusing instead on the semantic relationship between instance tables, trace tables, and constraints. This abstraction enables a precise classification of \textit{under-constrained} (soundness) and \textit{over-constrained} (completeness) bugs in the constraint systems of zkVM.

\subsection{Instance and Trace Tables}

Let $\mathbb{F}$ be a finite field used by the zkVM to represent scalar values. For a positive integer $w$, we write $(\mathbb{F}^w)^*$ to denote the set of matrices over $\mathbb{F}$ with $w$ columns and an arbitrary number of rows.

%%%REWRITE
We distinguish between {\it instance}, which encodes execution-invariant data such as the program binary and public inputs, and {\it trace tables}, which record execution-dependent data like the CPU table. For the purposes of this formalization, it is assumed that the given program binary represents a deterministic program. 

Let  $\mathbf{w}_{\mathrm{inst}} = (w_{s_1}, \dots, w_{s_k})$ and $\mathbf{w}_{\mathrm{trace}} = (w_{d_1}, \dots, w_{d_m})$ be the column widths of $k$ instance tables and $m$ trace tables, respectively. The spaces of instance and trace tables are then defined as:
\begin{align}
    \mathbb{S} &= \prod_{i=1}^{k} (\mathbb{F}^{w_{s_i}})^*, &
    \mathbb{D} &= \prod_{j=1}^{m} (\mathbb{F}^{w_{d_j}})^*.
\end{align}

\subsection{Constraint System}

A zkVM proves correctness of execution by enforcing algebraic constraints over trace tables, formally defined as follows.

\begin{definition}[AIR-Lookup-based zkVM Constraints]
Formally, a constraint system in a zkVM is a function $\mathbf{C}: \mathbb{S} \times \mathbb{D} \to \{\mathrm{True}, \mathrm{False}\}$, which is defined as
\begin{equation}
\begin{aligned}
&\mathbf{C}(\mathbf{T}_{\mathrm{inst}}, \mathbf{T}_{\mathrm{trace}}) = \\
&\displaystyle
\bigwedge_{A \in \mathcal{A}} A(\mathbf{T}_{\mathrm{trace}})
\;\land\;
\bigwedge_{L \in \mathcal{L}} L(\mathbf{T}_{\mathrm{inst}}, \mathbf{T}_{\mathrm{trace}}) \;\land\; PV(\mathbf{T}_{\mathrm{inst}}, \mathbf{T}_{\mathrm{trace}})
\end{aligned}
\end{equation}
Here, $\mathcal{A}$ denotes a set of {AIR constraints}, which enforce local transition rules between consecutive rows of trace tables, $\mathcal{L}$ denotes a set of {lookup constraints} enforcing global consistency across tables, and $PV$ is a {public value constraint} that directly inspects the public tables. 
\end{definition}

Note that while AIR and lookup constraints are used to generate cryptographic proofs, and therefore must satisfy certain algebraic formats, the public value constraint is an auxiliary function that directly checks properties of public tables and can be arbitrary.

\subsection{Bugs in zkVMs}
\label{subsec:bugs-in-zkvms}

To define bugs precisely, we reason about \emph{canonical} trace tables, which uniquely represent execution traces. This avoids multiplicities caused by padding, permutation, or other representational redundancies in trace tables. For example, a zkVM might pad a trace table with empty rows to match a power-of-two height. Both the unpadded and padded versions might satisfy constraints, but they represent the same execution trace. Our canonicalization ensures these are treated as the same trace table. The detailed modeling and construction of canonical trace tables is deferred to \S~\ref{appendix-canoincal}.

Let $\hat{\mathbb{D}} \subseteq \mathbb{D}$ denote the set of canonical trace tables, and we define a canonical solution set, which is the set of all canonical trace tables that satisfy the constraints for the fixed instance tables. 

\begin{definition}[Canonical Solution Set]
For a given instance table $\mathbf{T}_{\mathrm{inst}} \in \mathbb{S}$, we define its \textit{canonical solution set} as
\begin{align}
    \mathcal{Q}(\mathbf{T}_{\mathrm{inst}}) = 
    \{ \mathbf{T}_{\mathrm{trace}} \in \hat{\mathbb{D}} \mid \mathbf{C}(\mathbf{T}_{\mathrm{inst}}, \mathbf{T}_{\mathrm{trace}}) = \mathrm{True} \}.
\end{align}
\end{definition}

We first define the local completeness and soundness properties of the constraints in zkVMs, which are desired properties of the constraints of the zkVM for the fixed instance tables.

\begin{definition}[Local Completeness]
We say the constraint $\mathbf{C}$ in a zkVM is locally complete for the instance tables $\mathbf{T}_{\mathrm{inst}}$ if \begin{equation}
  \exists{\,\mathbf{T}_{\mathrm{trace}} \in {\mathbb{D}}}.\;\; \mathbf{C}(\mathbf{T}_{\mathrm{inst}}, \mathbf{T}_{\mathrm{trace}}) = \mathrm{True}
\end{equation}
\end{definition}

\begin{definition}[Local Soundness]
We say the constraint $\mathbf{C}$ in a zkVM is locally sound for the instance tables $\mathbf{T}_{\mathrm{inst}}$ if \begin{equation}
|Q(\mathbf{T}_{\mathrm{inst}})| \leq 1
\end{equation}
\end{definition}

Based on the above desired properties, we classify zkVM bugs according to the cardinality of the canonical solution set $\mathcal{Q}(\mathbf{T}_{\mathrm{inst}})$.

\begin{definition}[Over-Constrained zkVM]
We say the constraint $\mathbf{C}$ is locally \textit{over-constrained} for $\mathbf{T}_{\mathrm{inst}}$ if $
    |\mathcal{Q}(\mathbf{T}_{\mathrm{inst}})| = 0$.
\end{definition} \noindent \hideaki{Under our threat model, the execution trace generated by the VM executor for the program and input encoded in $T_{\mathrm{inst}}$ is a valid execution. Thus, $|Q(T_{\mathrm{inst}})| = 0$ means that
the constraint system rejects this valid execution trace, and is thus over-constrained for $T_{\mathrm{inst}}$.}

\begin{definition}[Under-Constrained zkVM]
We say the constraint $\mathbf{C}$ is locally \textit{under-constrained} for $\mathbf{T}_{\mathrm{inst}}$ if $|\mathcal{Q}(\mathbf{T}_{\mathrm{inst}})| > 1$.
\end{definition} \noindent This means that multiple distinct execution traces satisfy the constraint for the same program and inputs, enabling unsound proofs.

\subsection{Canonical Representation}
\label{appendix-canoincal}

To formally reason about the uniqueness of solutions in $Q(\mathbf{T}_{\mathrm{inst}})$, we must ensure that the trace tables in $\mathbb{D}$ map 1-to-1 with execution traces. In practice, a zkVM might introduce representational redundancies, such as padding rows or permuting non-deterministic lookup tables, which could artificially inflate the cardinality of $Q(\mathbf{T}_{\mathrm{inst}})$ without representing a true logic bug. We resolve this by defining a canonical space $\hat{\mathbb{D}} \subseteq \mathbb{D}$ that filters out such artifacts.

We first introduce the set of semantic execution records $\mathbb{E}$, and the canonicalizer $\mathcal{R}: \mathbb{D} \to \mathbb{E}$ that maps trace tables to a set of execution records. Intuitively, $\mathcal{R}$ replays the trace tables to obtain the set of execution records $\mathbb{E}$, abstracting away the redundancy. Specifically, we view that $\mathcal{R}$ is a union of canonicalizers for each table, i.e., $\mathcal{R}(\mathbf{T}_{\mathrm{trace}}) = \mathcal{R}_{\mathrm{cpu}}({T}_{\mathrm{cpu}}) \cup \mathcal{R}_{\mathrm{add}}({T}_{\mathrm{add}}) \cup \mathcal{R}_{\mathrm{mul}}({T}_{\mathrm{mul}}) \cup \cdots$.

For example, consider an ALU Add Table with columns for inputs ($\texttt{a}, \texttt{b}$), output ($\texttt{res}$), carry bits ($\texttt{c}$), and a selector ($\texttt{s}$). While the raw table $T_{\mathrm{Add}}$ includes all these columns with dummy padding rows (where $\texttt{s}=0$), the canonicalizer $\mathcal{R}$ acts as a semantic projection:
\begin{equation}
    \mathcal{R}_{\mathrm{Add}}(T_{\mathrm{Add}}) = \{ (\texttt{a}_i, \texttt{b}_i, \texttt{res}_i) \mid \texttt{row}_i \in T_{\mathrm{Add}} \land \texttt{s}_i = 1 \}
\end{equation}
By filtering out internal witnesses (e.g., $\texttt{c}$) and padding, and normalizing the permutation, $\mathcal{R}$ maps different raw representations of the same computation to a unique record set.

\begin{definition}[Semantic Equivalence of Trace Tables] 
We define a binary relation $\sim_{\mathcal{R}}$ on $\mathbb{D}$ derived from the canonicalizer $\mathcal{R}$. Two sets of dynamic traces $\mathbf{T}_1, \mathbf{T}_2 \in \mathbb{D}$ are semantically equivalent if and only if they reconstruct to the same execution record:

\begin{equation}
    \mathbf{T}_1 \sim_{\mathcal{R}} \mathbf{T}_2 \iff \mathcal{R}(\mathbf{T}_1) = \mathcal{R}(\mathbf{T}_2)
\end{equation}
\end{definition}

\noindent We then define the canonical trace tables space as the quotient set:

\begin{definition}[Canonical Trace Table Space]
\label{def:canonical-trace-table-space} The space of canonical trace tables, denoted $\hat{\mathbb{D}}$, is defined as the quotient set of $\mathbb{D}$ by the relation $\sim \mathcal{R}$:
\begin{equation}
    \hat{\mathbb{D}} := \mathbb{D} / \sim_{\mathcal{R}} = \{ [\mathbf{T}_{_{\mathrm{trace}}}]\sim \mid \mathbf{T}_{\mathrm{trace}} \in \mathbb{D} \}
\end{equation}, where $[\mathbf{T}_{_{\mathrm{trace}}}]\sim$ represents the equivalence class of traces indistinguishable from $\mathbf{T}_{\mathrm{trace}}$ by $\sim \mathcal{R}$.
\end{definition}

By construction, there exists a natural bijection between the canonical trace space $\hat{\mathbb{D}}$ and the image of the canonicalizer $\mathrm{Im}(\mathcal{R}) \subseteq \mathbb{E}$. This ensures that every set of canonical dynamic traces corresponds to a unique semantic execution trace.
\section{\sys: Design}

The primary challenge in testing and verifying zkVMs lies in the massive, non-linear constraint system $\mathcal{C}$ defined over multiple trace tables. \sys addresses this by shifting from individual-witness sampling to an integer interval lattice: modeling trace cells as intervals allows the branch-and-bound method to formally and efficiently certify whole regions of the search space at once.

\subsection{Localized Verification}
\label{subsec:lv}

Instead of attempting to verify the global correctness of the entire VM at once, we isolate a specific target table $T_{\mathrm{target}} \in \mathbb{D}_i$. In this localized mode, we assume the peripheral tables $\mathbf{T}_{\mathrm{rest}} = \mathbf{T}_{\mathrm{trace}} \setminus \{T_{\mathrm{target}}\}$ are semantically correct and consistent with the static instance tables $T_{\mathrm{inst}}$. This allows us to derive a localized constraint function $\mathcal{C}$ by partially evaluating $\mathbf{C}$ on the fixed witness: $\mathcal{C}(T_{\mathrm{target}}) = \mathbf{C}(\mathbf{T}_{\mathrm{inst}}, \mathbf{T}_{\mathrm{rest}} \cup \{T_{\mathrm{target}}\})$. This modularity allows \sys to bypass the complexity of global consistency proofs while exhaustively verifying the arithmetic correctness of individual units.

\subsection{The Trace Interval Lattice}

We model the search space for $T_{\mathrm{target}}$ as a lattice of integer intervals. This choice is motivated by the observation that complex zkVM constraints can be soundly approximated on intervals, enabling efficient and parallel traversal of the entire search space.

We first introduce the following notations:

\textbf{1) Integer Interval}: Each cell is abstracted as an integer interval, where a set of intervals is denoted as $\mathbb{I} = \{[a, b] \in \mathbb{Z} \times \mathbb{Z} \mid a \leq b\}$. %For intervals $I = [L, R]$ and $I' = [L', R']$, and we denote $I \subseteq I'$ if $L' \in L$ and $R \leq R'$.

\textbf{2) Interval Matrix} Let $\widetilde{T} \in \mathbb{I}^{H \times W}$ denote an interval matrix, where each cell $\widetilde{T}_{i, j}$ is filled with an integer interval $I_{i,j}$. This $\widetilde{T}$ serves as a search state in our \sys. 

Since the integer domain $\mathbb{Z}$ forms a lattice naturally with the partial order given by inequality, the meet by min, and the join by max, the set of all such states $\widetilde{T}$ itself forms a lattice~\cite{fernandez2006interval}, under the partial order of interval inclusion ($\widetilde{T} \sqsubseteq \widetilde{T}'$ if $\forall i,j.\; \widetilde{T}_{i, j} \subseteq \widetilde{T}'_{i, j}$). 

\textbf{3) Interval-to-Field Embedding}: To bridge the gap between the finite field and integer interval, we introduce the interval-to-field embedding of the integer interval $\llbracket \cdot \rrbracket_p : \mathbb{I} \to 2^{\mathbb{F}_p}$ by: $\llbracket I \rrbracket_p = \{x \bmod p \mid x \in \mathbb{Z} \;\, s.t. \;\, a \leq x \leq b\}$. This is also naturally extended to the embedding for a table: $\llbracket \widetilde{T} \rrbracket_p = \{ T \mid \forall i,j.\;\, T_{i, j} \in \llbracket \widetilde{T}_{i,j} \rrbracket_p \}$.

\subsection{Sound Over-approximation of Constraints}
\label{subsec:sound-approximation}

This section derives the sound over-approximation for the three classes of zkVM constraints: AIR, Lookups, and public value check.

\smallskip
\noindent \textbf{Soundness Condition.} For a localized zkVM constraint $\mathcal{C}: \mathbb{F}_p^{H \times W} \to \{\mathsf{True}, \mathsf{False}\}$, its sound over-approximation for an integer interval is formally defined as follows:

\begin{definition}[Sound Over-Approximation]
\label{def-soa}
$\widetilde{\mathcal{C}}: \mathbb{I}^{H \times W} \to \{\mathsf{True}, \mathsf{False}, \mathsf{Uncertain}\}$ is a sound over-approximation of the constraint $\mathcal{C}: \mathbb{F}_p^{H \times W} \to \{\mathsf{True}, \mathsf{False}\}$ if for any $\widetilde{T} \in \mathbb{I}^{H \times W}$, it holds:

\begin{enumerate}
\item $\widetilde{\mathcal{C}}(\widetilde{T}) = \mathsf{True} \implies \forall T \in \llbracket \widetilde{T} \rrbracket_p.\;\, \mathcal{C}(T) = \mathsf{True}$.
\item $\widetilde{\mathcal{C}}(\widetilde{T}) = \mathsf{False} \implies \forall T \in \llbracket \widetilde{T} \rrbracket_p.\;\, \mathcal{C}(T) = \mathsf{False}$.
\item $\widetilde{\mathcal{C}}(\widetilde{T}) = \mathsf{Uncertain}$ whenever neither (1) nor (2) can hold.
\end{enumerate}
\end{definition}

\noindent This ensures that any definitive result ($\mathsf{True}$ or $\mathsf{False}$) returned by $\widetilde{\mathcal{C}}$ is globally valid for all points in the sub-lattice.

\smallskip
\noindent \textbf{AIR Constraints.} AIR $\mathcal{A}$ consists of a set of polynomials $\{ \phi_1,\phi_2, \dots \}$ where each polynomial $\phi: \mathbb{F}^{2W} \to \mathbb{F}$ is defined for two consecutive rows. The satisfaction of $\mathcal{A}$ for a concrete table $T$ is defined as \begin{equation}
   \mathcal{A}(T) := \forall \phi \in \mathcal{A}.\;\,\forall i \in \{1,\dots, H-1\}.\;\, \phi(T_i, T_{i+1}) = 0
\end{equation} \noindent, where $T_i$ denotes the $i$-th row of the table $T$.

We then define the sound over-approximation $\tilde{\phi}$ using the basic interval arithmetic (see Appendix~\ref {appendix:subsec:ia} for the details). The satisfiability of AIR over $\widetilde{T}$ is over-approximated using the zero-satisfiability predicate $SAT_Z([L, H]) \iff \exists k \in \mathbb{Z}.\;\, kp \in [L, H]$:

\begin{align}\widetilde{\mathcal{A}}(\widetilde{T}) := \begin{cases}\mathsf{True} & \text{if } \forall \phi \in \mathcal{A}.\;\, \exists{k \in \mathbb{Z}}.\;\, \tilde{\phi}(\widetilde{T}) = [kp, kp] \\ \mathsf{False} & \text{if } \exists \phi \in \mathcal{A}.\;\, \neg SAT_Z(\tilde{\phi}(\widetilde{T})) \\ \mathsf{Uncertain} & \text{otherwise}\end{cases}\end{align}

\smallskip
\noindent \textbf{Lookup Constraints.} Lookup constraints ensure consistency between the target table $T$ and auxiliary tables. We categorize these into five types: Program, Byte, ALU, Control-Flow, and Memory. Detailed constructions are in Appendix~\ref{appendix:lookup}. For all types, the sound over-approximation $\widetilde{\mathcal{L}}$ returns $\mathsf{True}$ if the interval collapses to a valid singleton, $\mathsf{False}$ if the intervals are entirely disjoint from valid states, and $\mathsf{Uncertain}$ otherwise.

\textit{(L-1) Program}: $\mathcal{L}_{\mathrm{prog}}$ ensures every executed instruction (PC, opcode, operands) exists in the static program table $T_{\mathrm{prog}}$.

\textit{(L-2) Byte}: $\mathcal{L}_{\mathrm{byte}}$ enforces 8-bit range $[0, 255]$ and bitwise consistency (e.g., $\texttt{AND}, \texttt{XOR}$). $\widetilde{\mathcal{L}}_{\mathrm{byte}}$ uses bit-masking to track fixed/varying bits for tight interval bounds.

\textit{(L-3) ALU}: $\mathcal{L}_{\mathrm{alu}}$ validates word-level arithmetic $\diamond \in \{+, -, \times, \div, \dots\}$. $\widetilde{\mathcal{L}}_{\mathrm{alu}}$ leverages opcode-specific interval approximations $\tilde{\diamond}$ to detect arithmetic soundness violations.

\textit{(L-4) Control-Flow}: $\mathcal{L}_{\mathrm{cf}}$ verifies the transitions of program counters $\text{pc}' = \mathrm{cf}_{\circ}(\text{pc}, a, b, c)$. $\widetilde{\mathcal{L}}_{\mathrm{cf}}$ lifts the transition function to its interval-based counterpart $\widetilde{\mathrm{cf}_{\circ}}$.

\textit{(L-5) Memory}: $\mathcal{L}_{\mathrm{mem}}$ enforces global Read-After-Write consistency and zero-value initialization. $\widetilde{\mathcal{L}}_{\mathrm{mem}}$ models a memory by dynamically splitting address segments to track value intervals across the execution timeline.

\smallskip
\noindent \textbf{Public Value Check.} Public value constraint $PV$ inspects the public parts of tables, such as the final execution state, including the final program counter and exit code. Most of these constraints are simple equalities ($=$) or non-equalities ($\neq$), and they can be evaluated over the integer interval domain in the same way as AIR.

\subsection{Lattice Branch-and-Bound Search}
\label{subsec:bps}

Algo.~\ref{alg:bnb_search} systematically explores the lattice from the initial state $\widetilde{T}_{\text{root}}$ toward concrete traces. One of the natural choices of initial interval is $[0, p - 1]$, where $p$ is the prime order of the finite field $\mathbb{F}_p$.

\smallskip
\noindent \textbf{Basic Workflow.} The algorithm maintains a queue of search states. In each iteration, it \emph{\textit{1)} $\mathsf{SelectCellToSplit}(\widetilde{T})$} selects a state $\widetilde{T}$ and a "divisible" cell $\widetilde{T}_{i,j}$, where an interval $[a, b]$ is considered divisible if $a \neq b$; \emph{\textit{2)} $\mathsf{SplitInterval}(\widetilde{T}_{i,j})$} branches by splitting the interval $\widetilde{T}_{i,j}$ into $k$ smaller intervals $\{I_1, \dots, I_k\}$, and ;\emph{\textit{3)} $\widetilde{\mathcal{C}}$} evaluates the new states. If $\widetilde{\mathcal{C}}(\widetilde{T}) = \mathsf{False}$, the entire sub-lattice is discarded, efficiently pruning the search space. When $\widetilde{\mathcal{C}}(\widetilde{T}) = \mathsf{True}$, $\widetilde{T}$ is concretized into a set of concrete tables. Each table $T_i$ is added to the result set $Q$ only if its canonical representation $\mathcal{R}(T_i)$ has not been seen before.

\smallskip
\noindent \textbf{Bug Condition.} Any discrepancy in the cardinality of $Q(\widetilde{T})$ signals a critical vulnerability: an under-constrained bug is reported if $|Q(\widetilde{T})| > 1$, while an over-constrained bug is flagged if $|Q(\widetilde{T})|= 0$.

\smallskip
\noindent \textbf{Framework Guarantee.} The reliability of \sys rests on the convergence and completeness properties of the branch-and-bound framework applied to the interval lattice of the target table. By employing a sound approximation $\widetilde{\mathcal{C}}$, \sys ensures that any region of the search space pruned as False is mathematically guaranteed to contain no valid execution trace for the target component.

\begin{theorem}[Soundness and Completeness of Local Verification]\label{thm:sc}
Let $\mathcal{C}$ be a localized zkVM constraint system restricted to a target table $T_{\mathrm{target}}$, given fixed instance tables $\mathbf{T}_{\mathrm{inst}}$ and assumed valid peripheral trace tables $\mathbf{T}_{\mathrm{rest}}$. Then, Algo.~\ref{alg:bnb_search} terminates, and the resulting canonical solution set $Q$ satisfies:
\begin{itemize}
    \item \textbf{Local Soundness}: Every element $T \in Q$ is a
    concrete trace that satisfies the localized constraints:
    $\mathcal{C}(T) = \mathsf{True}$.
    \item \textbf{Local Completeness}: For any concrete $T$ with
    $\mathcal{C}(T) = \mathsf{True}$, there exists $T' \in Q$ such
    that $\mathcal{R}(T') = \mathcal{R}(T)$.
\end{itemize}
\end{theorem} \noindent See Appendix~\ref{subsec:proof} for the proof sketch.

\begin{algorithm}[!th]
\caption{Branch-and-Bound Lattice Search for zkVM}
\label{alg:bnb_search}
\begin{algorithmic}[1]
\Require Localized constraints $\mathcal{C}$ with respect to the given instance tables $\mathbf{T}_{\mathrm{inst}}$ and the target table $T_{\mathrm{target}}$
%\Ensure A set of canonical trace tabless $\mathcal{Q} \subset \widetilde{\mathbb{D}}$ satisfying $\mathcal{C}$

\State $\widetilde{T}_{\mathrm{root}} \gets$ Initialize all cells in $T_{\mathrm{target}}$ with interval $[0, p-1]$
\State $\mathsf{Queue}.\text{push}(\widetilde{T}_{\mathrm{root}})$, $\quad$ $\mathcal{Q} \gets \emptyset$ \Comment{Initialization}

\While{$\mathsf{Queue}$ is not empty}
    \State $\widetilde{T} \gets \mathsf{Queue}.\text{pop}()$
    \State $\mathsf{flag} \gets \widetilde{\mathcal{C}}(\widetilde{T})$ %\Comment{// Evaluate constraints using Interval Arithmetic}

    \If{$\mathsf{flag} = \mathsf{False}$} {continue} \Comment{Pruned}
    \ElsIf{$\mathsf{flag} = \mathsf{True}$} \Comment{Certain}
        \State $\{T_1, T_2, ..\} \gets \mathsf{Concretize}(\widetilde{T})$
        %\State $\mathcal{Q} \gets \mathcal{Q} \cup \{T_1, T_2, \dots\}$
        \State $\mathcal{Q} \gets \mathcal{Q} \cup \{\, T_i \mid T_i \in \{T_1, T_2, ..\}, \forall q \in \mathcal{Q},\; \mathcal{R}(q) \neq \mathcal{R}(T_i)\}$
    \Else \Comment{Uncertain}
        \State $(i,j) \gets \mathsf{SelectCellToSplit}(\widetilde{T})$ %\Comment{e.g., widest interval}
        \State $(\widetilde{T}_{left}, \widetilde{T}_{right}) \gets \mathsf{SplitInterval}(\widetilde{T}_{i,j})$
        %\State $p_{left} \gets \rho(S_{left}), \quad p_{right} \gets \rho(T_{right})$
        \State $\mathsf{Queue}.\text{push}(\widetilde{T}_{left})$, $\quad$ $\mathsf{Queue}.\text{push}(\widetilde{T}_{right})$
    \EndIf
\EndWhile \If{$|\mathcal{Q}| > 1$} Found Under-Constrained Bug \EndIf
\If{$|\mathcal{Q}| = 0$} Found Over-Constrained Bug \EndIf
%\Return $\mathcal{Q}$
\end{algorithmic}
\end{algorithm}

\subsection{Input Range Verification}
\label{subsec:rv}

We apply \rangeverification to instruction tables, such as ALU and Control-Flow tables, whose columns cleanly partition into inputs, output, and other auxiliary cells. Under a sound and complete constraint system, fixing the inputs determines a unique canonical row: local soundness forces at most one satisfying assignment to the output, and local completeness forces at least one.

To verify a range, we initialize the inputs with a broad interval $[L, R]$ rather than a singleton value. Because our branch-and-bound algorithm maintains a total cover of the search space, it systematically partitions this range until every valid execution trace is found and canonicalized. Let $\mathcal{I}_{in}$ be the set of inputs initialized with intervals, and let $V = \prod_{I \in \mathcal{I}_{in}} |I|$ be the total number of discrete input combinations within that range. If the constraint system is perfectly sound and complete, each unique input combination should correspond to exactly one semantically distinct execution trace. \sys evaluates the correctness of the range by comparing the cardinality of the resulting canonical solution set $|Q|$ with the expected volume $V$: sound and complete if $|Q| = V$; under-constrained if $|Q| > V$; over-constrained if $|Q| < V$.
\section{\sys: Implementation}

We implement \sys as a high-performance localized verification engine in approximately 4,500 lines of Rust. 

\subsection{Automated Constraint Extraction}
\label{subsec:ace}

\sys automates the extraction of arithmetic constraints from source code of Plonky3-based zkVMs using a dual-path architecture: a macro-based engine for AIR and static semantic lifting for lookups. The framework is highly extensible; adapting \sys to a new zkVM requires around 500 lines of tailored code.

\smallskip
\noindent \textbf{AIR Extraction via Rust Macros.} To handle diverse zkVM forks, \sys uses a unified Rust macro interface that intercepts symbolic execution. This automatically captures symbolic nodes (e.g., \texttt{Add}, \texttt{Mul}, \texttt{IsTransition}) directly from AIR definitions, ensuring compatibility across heterogeneous implementations.

\smallskip
\noindent \textbf{Lookup Extraction and Semantic Lifting.} \sys statically analyzes cross-table value transfers, which are implemented as channel-based virtual buses, and reconstructs their functional intent into semantically equivalent symbolic constraints  (Fig.~\ref{fig:channel}).

\smallskip
\noindent \textbf{Public Value and Extensibility.} Public value constraints, typically simple checks in standard Rust, are manually rewritten in our symbolic representation. 

\begin{figure}[!ht]
    \centering
    \includegraphics[width=0.9\linewidth]{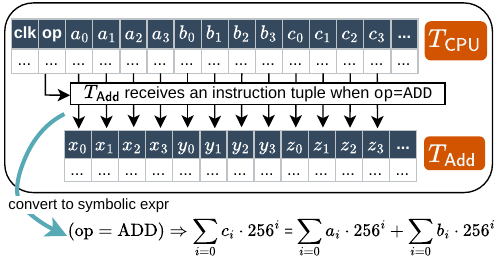}
    \caption{The CPU table offloads complex operation constraints to auxiliary tables such as the ADD table, by connecting their columns via a channel, and \sys converts them into semantically equivalent constraints.}
    \label{fig:channel}
\end{figure}

\subsection{Canonicalization}
\label{subsec:cano}

A significant challenge in zkVM verification is representational redundancy, where multiple distinct matrices may represent the same execution trace due to implementation-level artifacts. For instance, a zkVM might pad a trace table with empty rows to satisfy height requirements or utilize row permutations. Without normalizing these variations, a solver would report an inflated cardinality for $Q$, leading to false-positive under-constrained bug reports.

To address this, \sys utilizes a canonicalizer ($\mathcal{R}$) to project raw trace tables into unique sets of semantic execution records. This function acts as a semantic filter, ignoring internal helper columns and padding while normalizing permutations.

For tables representing each ALU and control-flow instruction, the canonicalized representation of a single row is implemented as a tuple of the current program counter, the next program counter, and all operands, which sufficiently captures the input and the output effect of each instruction (see Fig.~\ref{fig:canonicalizer} for the example pseudocode).

For tables that manage memory instructions such as \texttt{LW} and \texttt{SW}, the canonicalized representation of each row is a tuple of columns corresponding to accessed memory values and immediates.

The CPU table canonicalizer focuses on reconstructing the unified timeline of execution. Each row is projected into a tuple comprising the clock cycle, the program counter, and a record of memory writes performed during that cycle.

\smallskip
\noindent \textbf{Formalized correctness in Lean~4.}
We formalize the canonicalizer design in Lean~4 and verify its key correctness properties at the level of the algorithmic specification.
The development assumes that, for each zkVM, the trace-table generator records every execution step as a non-padding row, with no spurious or missing records.
Under this assumption, we prove three properties.

First, \textit{generator-independence}: any two generators satisfying this assumption produce the same canonical form for a given execution, ensuring that representational artifacts such as padding and row permutations are absorbed by $\mathcal{R}$.
Second, a constructive \textit{bijection} between the canonical trace space and the image of the canonicalizer, establishing that canonical representations correspond one-to-one with trace equivalence classes.
Third, \textit{injectivity} with respect to distinct executions, which underpins the cardinality criterion of \S~\ref{subsec:bugs-in-zkvms}: distinct executions map to distinct canonical forms, so $|Q|$ counts distinct executions rather than representational variants.

The framework is instantiated for the concrete column layouts of each evaluated zkVM; see Appendix~\ref{app:canonicalization} for the formal statements. %and examples.

\subsection{Search Optimization} 
\label{subsec:optimization}

The search space for zkVM trace tables is inherently high-dimensional, often involving thousands of cells in a table. To ensure that \sys can traverse this lattice within practical time limits, we implement several systems-level optimizations, including parallel execution, heuristic prioritization, and strategic initialization.

\smallskip
\noindent \textbf{Parallelization.} The search process is implemented as a multi-thread parallelized algorithm. Multiple worker threads concurrently pop candidate tables from a global queue, evaluate ($\tilde{\mathcal{C}}$), and execute recursive partitioning. If a thread determines that a sub-lattice is False, it immediately discards the entire branch and steals the next-highest-priority candidate from the queue.

\begin{figure}[!th]
\centering
\begin{align*}
\inferrule{I_s = [1, 1] \quad s \cdot (v - [\tau,\tau]) = 0}{I_v \leftarrow v \cap [\tau, \tau]} \quad (\textsc{Conditional-Const-Gate)}& \\ 
\frac{I_s = [1, 1] \quad s \cdot (a - b) = 0}{I_a, I_b \leftarrow a \cap b} \quad (\textsc{Conditional-Var-Gate})& 
\\
\inferrule{I_a = [a_{lo}, a_{hi}] \quad I_L = [\ell, h] \quad I_e =[e_{lo}, e_{hi}] \\\\ a = L - d \cdot e \quad d > 0 \quad a,e \notin \text{FV}(L) \\\\ |\ell| + |h| + d(|e_{lo}| + |e_{hi}|) + |a_{lo}| + |a_{hi}| < p}{I_e \leftarrow I_e \cap \left[ \left\lfloor \frac{\ell - a_{hi}}{d} \right\rfloor, \left\lfloor \frac{h - a_{lo}}{d} \right\rfloor \right]} \quad (\textsc{Affine-Backward})&
\end{align*}
\caption{Rules for Backward Refinement. $\text{FV}(L)$ is a set of free variables in an expression $L$.}
\label{fig:backward-refinment}
\end{figure}

\smallskip
\noindent \textbf{Backward Refinement.} To accelerate the convergence of the branch-and-bound search, \sys implements Backward Interval Refinement, which propagates constraints in reverse to shrink the search space, based on common patterns in real-world zkVMs: Conditional Constant-Gate, Conditional Variable-Gate, and Affine Backward Refinement. We denote the interval assigned for a variable $x$ by $I_x$, and the interval evaluated for an expression $L$ by $I_L$. See Fig.~\ref{fig:backward-refinment} for the formal inference rule and Appendix~\ref{appendix:implementation} for proofs.

\textit{(BR-1) Conditional Constant-Gate}: A common pattern in zkVM arithmetization is the use of a selector variable ($s$) to gate a specific constraint. For instance, a constraint might enforce that a register $v$ must equal a target constant $\tau$ only when $s=1$. \sys automatically detects this pattern and shrinks the interval of $v$ to its intersection $v \cap [\tau, \tau]$. If the intersection $v \cap [\tau, \tau]$ is empty, the branch is immediately pruned as unsatisfiable.

\textit{(BR-2) Conditional Variable-Gate}: Another popular pattern is conditional equality used to represent register-to-register moves or alias resolution. \sys replaces the intervals for both $a$ and $b$ with their intersection $a \cap b$, and prunes the state if they are disjoint.

\textit{(BR-3) Affine Backward Refinement}: The third common pattern is the affine relation. In many zkVMs, a value is decomposed into smaller limbs or quotients, and they are modeled as affine relations of the form $a = L - d \cdot e$, where $a$ is a remainder or low-limb variable, $L$ is an affine symbolic expression (independent of $a$ and $e$), $e$ is a quotient variable, and $d > 0$ is a constant stride. Given the current interval bounds for the affine expression $I_L = [\ell, h]$, the limb $I_a = [a_{min}, a_{max}]$, and the variables $a$ and $e$ do not occur free in $L$, we can derive tighter bounds for the quotient $e$.

\smallskip
\noindent \textbf{Certified Circuit Elision.} Many zkVMs use range-check circuits to verify that four-limb field elements fit within a specific range, like the BabyBear or KoalaBear range. Following best practices in formal verification, we simplify these by lifting them into dedicated IR nodes, such as \texttt{BabyBearRange} or \texttt{KoalaBearRange}. This allows the interval evaluator to apply immediate bounds to the variable, significantly reducing the complexity of the symbolic tree and accelerating the pruning process.

\smallskip
\noindent \textbf{Optimization for Detecting Under-Constrained Bugs.} In addition to the general optimizations that help the search terminate as early as possible, we introduce two techniques specifically designed to detect under-constrained bugs, which require finding satisfying assignments as quickly as possible.

\textit{Prioritization}:
Rather than exploring the lattice using a naive depth-first or breadth-first strategy, \sys uses a priority queue to \hideaki{direct the search toward regions that are more likely to contain valid but unintended execution traces. For each search state $\widetilde{T}$, we define its priority score as $\eta(\widetilde{T}) \triangleq \sum_{\tilde{c} \in \widetilde{\mathcal{C}}} \mathbbm{1}(\tilde{c}(\widetilde{T}) = \mathsf{True})$, i.e., the number of constraints that are already definitively satisfied over the corresponding interval region. We call a state \emph{almost valid} when $\eta(\widetilde{T})$ is close to $|\widetilde{\mathcal{C}}|$. Because an under-constrained bug manifests as a second, distinct satisfying trace, regions in which most constraints are already satisfied are heuristically more likely to yield such a trace after further refinement than regions in which fewer constraints are definitively satisfied. This score affects only the order in which states are explored and is never used to prune a region; therefore, it does not affect soundness or the set of states explored when the search runs to completion.}

\textit{Warm-Starting}: \sys can also leverage domain knowledge of the virtual machine. In particular, the search can be warm-started by initializing the lattice near the "honest" execution trace produced by the VM’s executor. This initialization allows the tool to quickly determine whether a second, conflicting witness exists by exploring the neighborhood of a known valid execution.

\section{Evaluation}

We conduct a comprehensive evaluation to answer the following research questions:

\noindent \textbf{RQ1:} How effectively can \sys locally verify real-world zkVMs at a given input point?

\noindent \textbf{RQ2:} How effectively can \sys locally verify real-world zkVMs over a given input range?

\noindent \textbf{RQ3:} How does increasing the number of parallel workers affect \sys's verification throughput?

\noindent \textbf{RQ4:} What is the contribution of each optimization to the verification performance of \sys?

\noindent \textbf{RQ5:} Can \sys uncover previously unknown bugs in real-world zkVM implementations?

\noindent \textbf{RQ6:} How does the bug-detection performance of \sys compare to existing fuzzers?

\smallskip \noindent \textbf{Benchmarks}. We evaluate \sys on five real-world zkVMs implemented using Plonky3: Pico~\cite{pico2025}, SP1~\cite{sp12026}, Sphinx~\cite{sphinx2025}, Valida~\cite{valida2025,valida2025core}, and Ziren~\cite{ziren2025}. These systems were selected based on their popularity, measured by the number of GitHub stars and forks. Pico, SP1, and Sphinx target the RISC-V architecture, Ziren targets MIPS, and Valida targets a custom instruction set architecture (ISA). Our study focuses on verifying and testing the core tables of these zkVMs, including the CPU table, ALU tables, control-flow tables, and memory instruction tables.

\smallskip \noindent \textbf{Configurations}. Experiments were conducted on an Intel Xeon 2.20GHz server with 31GB RAM running Ubuntu 22.04.3 LTS, utilizing four worker threads by default. We distinguish between \emph{instruction tables}, which constrain the input-output effects of specific opcodes (e.g., ALU, control-flow, and memory instruction tables), and \emph{non-instruction tables}, which manage global state and execution flow (e.g., the CPU table and global memory table).

Exhaustive verification analysis (RQ1–3) focused exclusively on instruction tables. We excluded non-instruction tables from this exhaustive scope because their high dimensionality (60–100 columns) and trace-dependent row counts make interval-based search computationally impractical for a full proof.

Instead, for bug detection (RQ5-6), we expanded our scope to include these non-instruction tables by leveraging warm-starting. \hideaki{For each table, we generate up to 30 short programs whose opcodes and operands are instantiated uniformly at random. Thus, for a zkVM with $n$ target tables, we use up to $30n$ honest seed traces in total. We execute each program using the target zkVM's native executor and table generator to obtain one honest seed trace. During the search, \sys randomly selects subsets of target columns and incrementally increases the subset size. The selected columns are lifted to interval domains, while the remaining cells retain their concrete witness values from the seed trace. This warm-starting strategy enables \sys to explore a localized neighborhood of valid executions and identify bugs in complex, large-scale tables.}

For single-point trials, we performed 200 independent runs per (table, opcode) pair using 32-bit random operands and a 100-second timeout. Range verification (RQ2) targeted the standard two-operand ($op\_b, op\_c$) instruction format by injecting a random interval of width $r \in \{1, 2, 8, 32, 128\}$ into one randomly selected limb per input. This configuration yields a search volume proportional to $r^2$ to evaluate solver convergence under uncertainty. \hideaki{Scalability (RQ3) reuses exactly the single-point opcode-wise benchmark instances from RQ1, with the same operands, random seeds, and timeout; only the number of worker threads $W\in\{1,2,3,4\}$ is varied.} We compared \sys against two baselines: the Z3 SMT solver (using QF\_BV logic via SMT-LIB2) and Arguzz, a mutation-based fuzzer we extended to support our five target zkVMs (see Appendix~\ref{appendix:experiemnt-details} for the detailed configurations for those baselines).

% \usepackage{booktabs}
% \usepackage{multirow}
\begin{table*}[]
\centering
\caption{Local verification performance at a single input point for \sys and the Z3 SMT solver. $N$ is the number of tables or opcodes. Table-wise, $N_v$ counts fully verified tables and Succ\% is the fraction of verified opcodes averaged over tables; opcode-wise, $N_v$ and Succ\%
denote verified opcodes and $N_v/N$. Time is the average verification time. $N_{c}$ is the number of instances verified by both, and Spd is the average speedup of \sys over Z3 on this common set. 
\sys achieves higher coverage ($+16.5$ pp) and up to $51.5\times$ speedup.}
\begin{tabular}{@{}c|ccccccccc|ccccccccc@{}}
\toprule
\multirow{3}{*}{zkVM} & \multicolumn{9}{c|}{Table-Wise} & \multicolumn{9}{c}{Opcode-Wise} \\ \cmidrule(l){2-19} 
 & \multicolumn{1}{c|}{\multirow{2}{*}{$N$}} & \multicolumn{3}{c|}{\sys} & \multicolumn{3}{c|}{SMT} & \multirow{2}{*}{$N_{c}$} & \multirow{2}{*}{Spd} & \multicolumn{1}{c|}{\multirow{2}{*}{$N$}} & \multicolumn{3}{c|}{\sys} & \multicolumn{3}{c|}{SMT} & \multirow{2}{*}{$N_{c}$} & \multirow{2}{*}{Spd} \\ \cmidrule(lr){3-8} \cmidrule(lr){12-17}
 & \multicolumn{1}{c|}{} & $N_{v}$ &  Succ\% & \multicolumn{1}{c|}{Time} & $N_{v}$ & Succ\% & \multicolumn{1}{c|}{Time} &  &  & \multicolumn{1}{c|}{} & $N_{v}$ & Succ\% & \multicolumn{1}{c|}{Time} & $N_{v}$ & Succ\% & \multicolumn{1}{c|}{Time} &  &  \\ \midrule
Valida & \multicolumn{1}{c|}{7} & 2 & 33.3\% & \multicolumn{1}{c|}{0.010s} & 1 & 16.7\% & \multicolumn{1}{c|}{0.593s} & 0 & {NA}\footnotemark & \multicolumn{1}{c|}{11} & 2 & 18.2\% & \multicolumn{1}{c|}{0.010s} & 2 & 18.2\% & \multicolumn{1}{c|}{0.593s} & 0 & NA \\
Sphinx & \multicolumn{1}{c|}{7} & 4 & 57.1\% & \multicolumn{1}{c|}{0.109s} & 2 & 28.6\% & \multicolumn{1}{c|}{1.429s} & 2 & 4.7x & \multicolumn{1}{c|}{18} & 10 & 55.6\% & \multicolumn{1}{c|}{0.158s} & 7 & 38.9\% & \multicolumn{1}{c|}{1.630s} & 7 & 22.3x \\
Pico & \multicolumn{1}{c|}{8} & 4 & 50.0\% & \multicolumn{1}{c|}{0.095s} & 2 & 25.0\% & \multicolumn{1}{c|}{1.117s} & 2 & 4.0x & \multicolumn{1}{c|}{25} & 10 & 40.0\% & \multicolumn{1}{c|}{0.144s} & 7 & 28.0\% & \multicolumn{1}{c|}{1.274s} & 7 & 27.5x \\
SP1 & \multicolumn{1}{c|}{10} & 6 & 60.0\% & \multicolumn{1}{c|}{0.070s} & 3 & 33.3\% & \multicolumn{1}{c|}{4.621s} & 3 & 82.4x & \multicolumn{1}{c|}{32} & 16 & 50.0\% & \multicolumn{1}{c|}{0.100s} & 10 & 31.2\% & \multicolumn{1}{c|}{2.497s} & 9 & 56.8x \\
Ziren & \multicolumn{1}{c|}{12} & 8 & 66.7\% & \multicolumn{1}{c|}{0.239s} & 2 & 28.2\% & \multicolumn{1}{c|}{4.516s} & 2 & 1.6x & \multicolumn{1}{c|}{41} & 23 & 56.1\% & \multicolumn{1}{c|}{0.246s} & 14 & 34.1\% & \multicolumn{1}{c|}{2.976s} & 8 & 92.2x \\ \midrule
All & \multicolumn{1}{c|}{44} & 24 & 55.8\% & \multicolumn{1}{c|}{0.132s} & 10 & 27.2\% & \multicolumn{1}{c|}{3.339s} & 9 & 29.8x & \multicolumn{1}{c|}{127} & 61 & 48.0\% & \multicolumn{1}{c|}{0.169s} & 40 & 31.5\% & \multicolumn{1}{c|}{2.204s} & 31 & 51.5x \\ \bottomrule
\end{tabular}
\label{tab:rq1}
\end{table*}

\subsection{RQ1: Local Verification Performance}

Table~\ref{tab:rq1} reports local verification performance under two aggregations. {Opcode-wise} treats each instruction opcode as a distinct target (e.g., \texttt{ADD}, \texttt{BEQ});  {table-wise} groups opcodes that share a single trace table  (e.g., all bitwise operations folded into one \texttt{Bitwise} table) into one entry. We use opcode-wise figures as the headline figures and report table-wise figures as a robustness check.

\smallskip \noindent \textbf{Overall Performance.} At the opcode level, \sys verified $48.0\%$ ($61/127$) of instruction targets, against Z3's $31.5\%$ ($40/127$), a $+16.5$\,pp gain. Aggregating by table widens the gap to $+28.6$\,pp ($55.8\%$ vs.\ $27.2\%$), confirming that \sys's advantage is not an artifact of how opcodes are bucketed. On the $31$ operations that both tools certify, \sys is on average $51.5\times$ faster.

\smallskip \noindent \textbf{Per-zkVM Analysis.} \sys outperforms Z3 on every evaluated zkVM. The opcode-wise gap is largest on Ziren ($+22.0$\,pp) and SP1 ($+18.8$\,pp), and remains positive on Sphinx ($+16.7$\,pp) and Pico ($+12.0$\,pp); on Valida the two tools tie at $18.2\%$ opcode-wise but \sys leads $33.3\%$ vs.\ $16.7\%$ table-wise. Valida's lower absolute coverage is consistent with its bit-decomposition-heavy design, in which several tables use up to $4\times$ as many columns as their counterparts in other zkVMs; the resulting blow-up inflates the search space for any solver operating cell-by-cell. The largest table-wise speedup, $82.4\times$ on SP1, is driven by its jump table: \sys certifies \texttt{JAL} and \texttt{JALR} in $\approx5$\,ms each, whereas Z3 spends $\approx1.3$\,s.

\smallskip \noindent \textbf{Failure Analysis.} Successful verifications included simple ALU tables, multiplication, shifts, and jumps. However, deep non-linearities, such as division and multi-byte stores, remained beyond the 100s timeout. Notably, Z3 solved the division table in Valida that \sys could not, as bit-vector theory provides formal decision procedures that interval approximations lack, suggesting that hybrid interval-SMT approaches could further improve the zkVM verification.

\footnotetext{The two operations in Valida verified by \sys and Z3 belong to different sets without any overlap.}

\begin{tcolorbox}[colback=gray!8, colframe=black, boxrule=0.4pt, arc=2pt, left=4pt, right=4pt, top=3pt, bottom=3pt]
\textbf{Answer to RQ1:} \sys outperforms SMT verification in both efficiency and coverage, achieving a {48.0\% success rate} (a {16.5 percentage point (pp) improvement} over Z3) and a {$51.5\times$ speedup}, establishing a new performance benchmark for automated ZK constraint analysis.
\end{tcolorbox}

\subsection{RQ2: Range Verification Performance}

\begin{table}[!th]
\centering
\caption{Speedup of \sys's \rangeverification over the iterative point-wise baseline, across input volumes $|V|$. %The baseline cost is $|V|$ times the mean single-point verification time from Tab.~\ref{tab:rq1} (i.e., the cost of invoking \sys independently at every point in the range).
}
\label{tab:rq2}
\begin{tabular}{c|cccc}
\toprule
    zkVM & $|V|=2^2$ & $|V|=2^6$ & $|V|=2^{10}$ & $|V|=2^{14}$ \\ \hline
Valida & 3.9$\times$ & 30.1$\times$ & 49.7$\times$ & 63.1$\times$ \\
Sphinx & 2.8$\times$ & 18.3$\times$ & 35.1$\times$ & 42.1$\times$ \\
Pico   & 3.0$\times$ & 27.3$\times$ & 39.5$\times$ & 44.2$\times$ \\
SP1    & 1.7$\times$ & 10.0$\times$ & 28.4$\times$ & 34.1$\times$ \\
Ziren  & 2.5$\times$ & 16.2$\times$ & 31.2$\times$ & 39.7$\times$ \\
 \hline
\end{tabular}
\end{table}

Tab.~\ref{tab:rq2} reports the speedup of \sys's \rangeverification over a theoretical iterative point-wise verification baseline, where the baseline cost for a volume $|V|$ is computed as $|V|$ times the mean single-point verification time from RQ1. Even at the smallest volume $|V| = 2^2$, \sys already outperforms point-wise enumeration, because a single interval-lattice traversal covers multiple discrete inputs with shared pruning. As $|V|$ grows, the gap widens rapidly: at $|V| = 2^{14}$, \sys achieves $34.1\times$ on SP1, $39.7\times$ on Ziren, $42.1\times$ on Sphinx, $44.2\times$ on Pico, and $63.1\times$ on Valida. This sublinear growth in verification time with respect to volume is the core benefit of \rangeverification: branch-and-bound collapses large contiguous regions into a single sound judgment, whereas iterative point-wise verification pays the full verification cost at every point.

Fig.~\ref{fig:range_sweep} shows this effect on the ADD table across all zkVMs, plotting \sys's measured verification time against the corresponding theoretical iterative point-wise verification cost as the volume increases. The solid curves lie strictly below the dashed iterative point-wise verification references, and the vertical gap widens as volume grows, confirming that \sys verifies orders of magnitude more inputs than point-wise testing would allow; equivalently, for a fixed target volume, \sys reaches it substantially faster. See Appendix~\ref{appendix:smt:range} for the SMT solver baseline for range verification, where the Z3 SMT solver exhibits a scalability bottleneck.

\begin{figure*}[!th]
    \centering
    \includegraphics[width=0.99\linewidth]{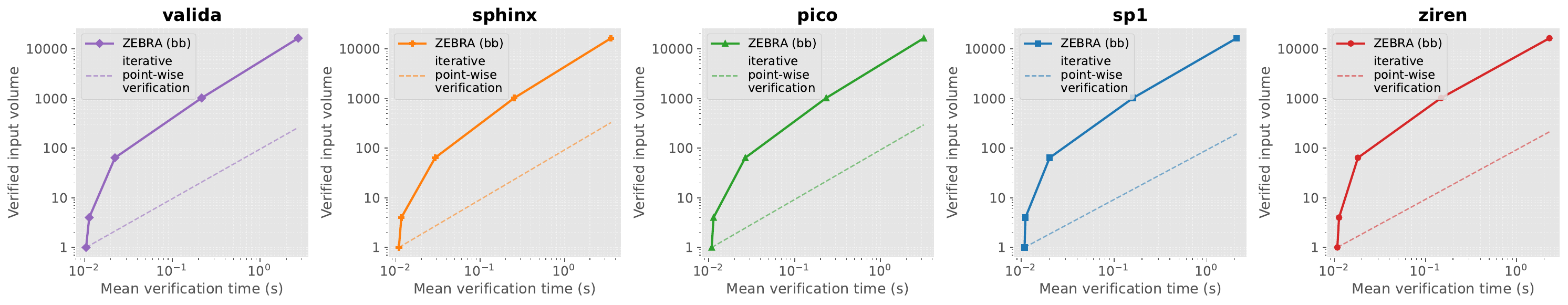}
\caption{Verified input volume (y) versus mean verification time (x)
for \sys on the ADD table. The widening vertical gap at larger volumes indicates that interval-lattice reasoning amortizes increasingly well as the range
grows. Z3 SMT solver cannot certify any of these tables even for a singleton input.}
    \label{fig:range_sweep}
\end{figure*}

\begin{tcolorbox}[colback=gray!8, colframe=black, boxrule=0.4pt, arc=2pt, left=4pt, right=4pt, top=3pt, bottom=3pt]
\textbf{Answer to RQ2.} \sys achieves substantial amortized efficiency by processing input domains as unified interval lattices rather than discrete points, yielding up to a $63\times$ speedup over the theoretical iterative point-wise verification baseline.
\end{tcolorbox}

\subsection{RQ3: Parallel Search Scalability}

Tab.~\ref{tab:worker} shows speedups from $W=1$ baseline. Ziren drops from 0.521s to 0.246s (2.12$\times$) at $W=4$, while SP1, Sphinx, and Pico achieve 1.24–1.63$\times$, showing that \sys effectively distributes independent sub-lattices across workers.

\begin{table}[!th]
\centering
\caption{Impact of parallelization on verification throughput with respect to the number of parallel worker threads $W$}
\label{tab:worker}
\begin{tabular}{c|cccc}
\toprule
zkVM & $W=1$ & $W=2$ & $W=3$ & $W=4$ \\ \hline
Valida & 0.011 & 0.011 & 0.011 & 0.010 \\
Sphinx & 0.257 & 0.168 & 0.166  & 0.158 \\
Pico & 0.179 & 0.156 & 0.150 & 0.144 \\
SP1 & 0.148 & 0.123 & 0.105 & 0.100 \\
Ziren & 0.521 & 0.271 & 0.260 & 0.246 \\
 \hline
\end{tabular}
\end{table}

\begin{tcolorbox}[colback=gray!8, colframe=black, boxrule=0.4pt, arc=2pt, left=4pt, right=4pt, top=3pt, bottom=3pt]
\textbf{Answer to RQ3.} The multi-threaded parallelized algorithm significantly enhances \sys's throughput, providing up to a 2.12x speedup with 4 workers.
\end{tcolorbox}

\subsection{RQ4: Ablation Study of Optimizations}

\begin{table}[!th]
\centering
\caption{Impact of Backward Refinement (BR) and Certified Circuit Elision (CCE). Suc\% is opcode-wise average. BR and CCE contribute to coverage by 24.4 and 11.8 pp, respectively.}
\begin{tabular}{c|ccc|ccc|ccc}
\hline
zkVM  & \multicolumn{3}{c|}{\sys} & \multicolumn{3}{c|}{\begin{tabular}[c]{@{}c@{}}\sys \\ w.o. BR\end{tabular}} & \multicolumn{3}{c}{\begin{tabular}[c]{@{}c@{}}\sys \\ w.o.  CCE\end{tabular}} \\ \hline
          & \#T & \#O & Suc\%   & \#T & \#O & Suc\%   & \#T & \#O & Suc\%   \\ \hline
Valida &   2     & 2     & 18.2\%  & 2     & 2     & 18.2\%  & 2     & 2     & 18.2\%  \\ 
Sphinx &   4     & 10    & 55.6\%  & 2     & 5     & 27.8\%  & 4     & 10    & 55.6\%  \\
Pico   &   4     & 10    & 40.0\%  & 2     & 5     & 20.0\%  & 4     & 10    & 40.0\%  \\
SP1    &   6     & 16    & 50.0\%  & 3     & 7     & 21.9\%  & 5     & 12    & 37.5\%  \\
Ziren  &   8     & 23    & 56.1\%  & 4     & 11    & 26.8\%  & 5     & 12    & 29.3\%  \\
\hline
Total  &  24    & 61    & 48.0\%  & 13    & 30    & 23.6\%  & 20    & 46    & 36.2\%  \\ \bottomrule
\end{tabular}
\label{tab:ablation}
\end{table}

Tab.~\ref{tab:ablation} shows the ablation result, where \#T and \#O show the number of locally verified tables and instruction opcodes, respectively. We conduct an ablation study to quantify the contributions of Backward Refinement (BR) and Certified Circuit Elision (CCE). BR is the primary driver of verification coverage: disabling it drops coverage from 48.0\% to 23.6\% (-31 operations), with the largest losses on Ziren (-12), SP1 (-9), Sphinx (-5) and Pico (-5). CCE is more architecture-specific, reducing coverage to 36.2\% (-11.8 pp, -15 operations) when disabled. The effect is strongest on Ziren (-11) and SP1 (-4), whose AIR constraints contain selector-multiplied sub-expressions that CCE simplifies; Pico is unaffected, as its constraint structure is already relatively canonical. CCE also improves efficiency, cutting mean verification time by up to 42\% on Ziren and SP1.

\begin{tcolorbox}[colback=gray!8, colframe=black, boxrule=0.4pt, arc=2pt, left=4pt, right=4pt, top=3pt, bottom=3pt]
\textbf{Answer to RQ4:} BR and CCE form a synergistic pipeline: BR contributes a 24.4 pp coverage increase via search-space reduction, and CCE adds up to 11.8 pp through constraint simplification. Neither alone matches the full \sys pipeline.
\end{tcolorbox}

\subsection{RQ5: Previously Unknown Bugs}

\sys identified 13 under- and over-constrained behaviors across the evaluated zkVMs (see Tab.~\ref{tab:bugs}). Of these, 11 are previously unknown bugs discovered via warm-started bug-finding mode; the remaining 2 (IDs 11 and 13) are intentional design choices surfaced during exhaustive verification. We classify severity into five categories: Critical affects arbitrary programs; High affects specific instruction families; Medium provides limited exploitability; Design Flaw denotes ISA-vs-constraint mismatches that occur only for corrupted programs; and Design Choice denotes behaviors structurally required to preserve constraint well-formedness.

ID 1 is a critical under-constrained bug in Valida where the $\texttt{is\_beq}$ and $\texttt{is\_bne}$ CPU flags are not linked to the opcode. Due to the lack of those links, a prover can toggle them to manipulate program counter transitions. IDs 2, 3, and 7 allow premature termination. In Valida, CPU constraints accept a matrix of padding rows, enabling empty executions. In Sphinx, termination is checked via $\texttt{next\_pc} = 0$, but this is bypassed for sufficiently short traces. IDs 4 and 5 further show soundness issues: branch instructions do not verify the second operand address, enabling substitution of comparison values, and the Memory table does not enforce that initial reads return zero, allowing arbitrary initial memory.

IDs 8, 9, 10, and 12 are classified as Design Flaw, reflecting gaps between VM logic and constraints due to corrupted programs. All of them are due to program counter overflow over a finite field, where the counter wraps around to zero and falsely signals termination, despite the assumption that only $\texttt{HALT}$ sets $\texttt{next\_pc}=0$. This behavior can happen only when the counter can be extremely large, and developers noted that these systems assume well-formed compiler-generated binaries, making such edge cases rare; thus, these behaviors remain unfixed in favor of performance.

IDs 11 and 13 are classified as Design Choice because the over-constrained behavior is structurally necessary: enforcing $\texttt{next\_pc} < p$ in the jump table prevents silent field
wrap-around, which would otherwise enable the PC-overflow bugs. We nonetheless label these as over-constrained rather than intended behavior as the zkVM executors themselves execute such jumps without error, and the rejection discards traces that are semantically valid executions. However, this reflects the zkVM designers' intention to trade completeness on these edge cases for prover performance, as supporting targets $\geq p$ would require an additional range-decomposition machinery.

\hideaki{Note that our findings are not tied to specific hand-crafted programs or inputs. Eight findings (IDs 1–3, 7–10, and 12) are reachable from any suitable warm-start program and are independent of its concrete contents. IDs 4–6 require only the presence of a relevant instruction family, such as branch or memory operations, rather than a particular program or input. IDs 11 and 13 were found during exhaustive verification rather than warm-started bug finding and therefore do not depend on the RQ5 trace corpus.}

% \usepackage{booktabs}
\begin{table*}[!th]
\centering
\caption{Under- and over-constrained behaviors discovered by \sys.
\textbf{UC}~=~under-constrained; \textbf{OC}~=~over-constrained.
Severity: \textcolor{red!70!black}{\textbf{Critical}}: affects arbitrary programs;
\textcolor{orange!80!black}{\textbf{High}}: affects specific instruction families with broad consequences;
\textcolor{yellow!50!black}{\textbf{Medium}}: affects specific instruction families with limited exploitability;
\textcolor{black!60}{\textbf{Design Flaw}}: affects only malformed programs.
\textcolor{blue!60}{\textbf{Design Choice}}: structurally required to preserve constraint well-formedness.
}
\label{tab:bugs}
\begin{tabular}{@{}c|c|c|c|l|c@{}}
\toprule
 \textbf{zkVM} & \textbf{ID} & \textbf{Type} & \textbf{Severity} & \textbf{Description} & \textbf{Status}  \\ \midrule

\multirow{6}{*}{Valida} 

& 1 &  UC & \Crit & \begin{tabular}[c]{@{}l@{}}  Opcode flags such as \texttt{is\_beq} / \texttt{is\_bne} are not properly linked to the opcode; \\ suppressing these flags allows malicious provers to alter control flow.\\ \end{tabular} & Fixed \\
\cmidrule(l){3-6}

& 2 &UC &   \Crit & \begin{tabular}[c]{@{}l@{}}  The CPU table does not constrain \texttt{is\_real\,=\,1} on the first row; \\ a prover can produce an accepted proof of empty execution for any program.
\end{tabular} & Fixed \\
\cmidrule(l){3-6}

& 3 & UC & \Crit &\begin{tabular}[c]{@{}l@{}}  The Final instruction is not constrained to be \texttt{STOP} in any verification stage; \\ execution may halt arbitrarily for any program with an accepted proof.\end{tabular} & Fixed \\
\cmidrule(l){3-6}

& 4 & UC & \High & \begin{tabular}[c]{@{}l@{}}  Branch instructions do not verify the second operand when it is an address; \\ a prover can substitute an arbitrary address to alter the comparison value. \end{tabular} & In contact \\
\cmidrule(l){3-6}

& 5 &UC &  \High & \begin{tabular}[c]{@{}l@{}}  The \texttt{Memory} table does not constrain zero-initialized cells to hold zero; \\ a prover can supply non-zero initial values to any cell.\\  \end{tabular} & Confirmed \\
\cmidrule(l){3-6}

& 6 & UC & \Medium & \begin{tabular}[c]{@{}l@{}}  Written values lack a range check; distinct byte decompositions evaluate \\ identically over the finite field, allowing the modification of written values.
\end{tabular} & In contact \\
\midrule

\multirow{2}{*}{Sphinx} 

& 7 & UC & \Crit & \begin{tabular}[c]{@{}l@{}} When the execution trace is sufficiently short, program termination is not \\ checked, allowing a malicious prover to terminate prematurely. \end{tabular} & In contact \\
\cmidrule(l){3-6}

& 8 & UC & \Obs & \begin{tabular}[c]{@{}l@{}} only \texttt{next\_pc\,=\,0} is required to show the normal termination; counter overflow \\ over finite field can satisfy this condition without executing a halt instruction.
\end{tabular} & In contact \\  
\midrule

Pico & 9 & UC & \Obs & \begin{tabular}[c]{@{}l@{}} The verifier uses \texttt{next\_pc\,=\,0} as proof of normal termination by a halt syscall; \\ counter overflow produces the same
    field value, falsely signaling a clean halt. \end{tabular}  & Confirmed \\ 
\midrule

\multirow{2}{*}{SP1} 

 & 10 & UC & \Obs & \begin{tabular}[c]{@{}l@{}} Both \texttt{next\_pc\,=\,0} and \texttt{exit\_code\,=\,0} checks the successful termination, \\ yet neither condition distinguishes a halt syscall
    from counter overflow.
    \end{tabular} & In contact \\
\cmidrule(l){3-6}

 & 11 & OC & \Info & \begin{tabular}[c]{@{}l@{}} \texttt{Jump} table enforces \texttt{next\_pc}\,$<$\,$p$ where $p{=}2^{31}{-}2^{27}{+}1$; jump targets ${\geq}\,p$ are \\
    unrepresentable in a valid trace, causing the table to reject it.
\end{tabular} & Known \\ 
\midrule

\multirow{2}{*}{Ziren} 

 & 12 & UC & \Obs & \begin{tabular}[c]{@{}l@{}}
The same \texttt{next\_pc\,=\,0} check is applied over
    KoalaBear ($p{=}2^{31}{-}2^{24}{+}1$); \\ counter overflow satisfies the
    halt condition without a halt syscall.
\end{tabular} & Confirmed \\ 
\cmidrule(l){3-6}

 & 13 & OC & \Info & \begin{tabular}[c]{@{}l@{}} \texttt{Jump} table enforces \texttt{next\_pc}\,$<$\,$p$ where $p{=}2^{31}{-}2^{24}{+}1$; jump targets ${\geq}\,p$ are \\
    unrepresentable in a valid trace, causing the table to reject it.
\end{tabular}  & Known\\  
 \bottomrule
\end{tabular}
\end{table*}

\begin{tcolorbox}[colback=gray!8, colframe=black, boxrule=0.4pt, arc=2pt, left=4pt, right=4pt, top=3pt, bottom=3pt]
\textbf{Answer to RQ5.} \sys identified 11 previously unknown bugs across five real-world zkVMs, plus 2 known design choices, for 13 total flagged behaviors.
\end{tcolorbox}

\subsection{RQ6: Comparison with Fuzzer}

\smallskip \noindent \textbf{Overall Detection Efficiency}. We evaluate \sys’s bug-detection capability against Arguzz, a state-of-the-art mutation-based zkVM fuzzer. Arguzz was run for six hours per target but failed to identify any vulnerabilities in these benchmarks, consistent with prior limitations. In contrast, \sys flagged 13 under- and over-constrained behaviors, 11 of them previously unknown bugs, with a 100-second budget applied per table (on average 900 seconds per zkVM) rather than per discovered bug.

\smallskip \noindent \textbf{Root Cause Analysis}. \sys and Arguzz both target soundness and completeness bugs in zkVM constraints. For under-constrained detection, however, \sys explores a strictly larger search space than Arguzz due to their structural gap. The zkVM prover pipeline composes two stages: an executor $E$ mapping a program to an execution trace, and a table generator $G$ lifting the trace into full tables $\mathbf{T} \in \mathbb{D}$ by populating auxiliary columns such as selector flags, carry bits, and padding witnesses. Arguzz mutates $E$ itself into a faulty executor $E'$, so the tables it tests are of the form $G(E'(P))$. \sys searches directly over $\mathbb{D}$. Because $G$ is a deterministic function, Arguzz's reachable tables are contained in $\mathrm{Im}(G \circ E') \subsetneq \mathbb{D}$. Any under-constrained bug whose exploit requires an auxiliary column to take a value $G$ would not produce lies outside this image and is undetectable by executor-level mutation under any time budget; $ G$ does not bind a real malicious prover who writes $T$ directly.

Bug ID~1 (Valida) illustrates this. The \texttt{is\_beq} and \texttt{is\_bne} flags are set by $G$ as pure functions of the opcode, so Arguzz cannot observe them as independent variables. The constraint system fails to tie them back to the opcode, so a prover writing $T$ directly can set \texttt{is\_beq} $= 1$ on a non-BEQ instruction and hijack control flow. The bug lives exactly in the gap between $\mathrm{Im}(G \circ E')$ and $\mathbb{D}$, and the executor-mutation fuzzers cannot catch it.

\begin{tcolorbox}[colback=gray!8, colframe=black, boxrule=0.4pt, arc=2pt, left=4pt, right=4pt, top=3pt, bottom=3pt]
\textbf{Answer to RQ6.} \sys outperforms mutation-based fuzzers by operating directly on the table domain $\mathbb{D}$ rather than the executor-reachable state spaces. While Arguzz found no bugs even after 6 hours/zkVM, \sys detected 13 under- and over-constrained behaviors, including 11 zero-day bugs in 900 seconds/zkVM.
\end{tcolorbox}
\section{Related Work}

The correctness of zkVM constraint systems has become an active area of research as these systems gain adoption in production. Prior work falls into three broad categories: dynamic testing, SMT-based automated verification, and interactive theorem proving.

\smallskip \noindent \textbf{Dynamic Testing.} Traditional testing methodologies for zkVMs rely on point-wise sampling to surface logic discrepancies. Production zkVMs such as SP1 employ extensive negative unit tests that inject random mutations into the proof generation process to uncover under-constrained bugs. zkFuzz~\cite{takahashi2025zkfuzz} introduced the first general fuzzing framework for ZK circuits based on code mutation, and Arguzz~\cite{hochrainer2025arguzz} extended these ideas to zkVMs by combining metamorphic testing with fault injection into the virtual machine's executor. While effective at uncovering shallow bugs, these dynamic approaches provide no theoretical guarantee of completeness and are structurally unable to reach tables whose helper columns are populated only during witness generation.

\smallskip \noindent \textbf{SMT-based Verification.} Another direction is SMT solvers to verify ZK constraint systems automatically. Picus~\cite{pailoor2023automated} combines structural analysis with SMT-based reasoning to detect under-constrained bugs in ZK circuits, and has been extended to SP1~\cite{veridise-sp1-picus} and Ziren~\cite{audithub-ziren-picus}. These deployments, while promising, have three structural limits, and \hideaki{they are not directly comparable to our evaluated Z3 baseline and cannot serve as drop-in experimental baselines for \sys}: Picus targets determinism, so over-constrained behaviors are out of scope; published deployments verify only small sub-chips, as scaling is blocked by Plonky3's lack of modular constraint blocks and I/O annotations~\cite{veridise-sp1-picus}; and Picus's circuit model does not natively handle lookup or permutation constraints~\cite{ke2025consistency}. ZIVER~\cite{ke2025consistency} addresses the last point with an inductive consistency-checking algorithm for SP1 that supports lookups, but requires manual translation to a custom DSL and covers only 7 small sub-components. Finally, while finite-field–aware SMT theories such as cvc5 have recently emerged, they still lack support for constraint patterns common in zkVMs, such as disequality ($\neq$) constraints used in public value checks.

\smallskip \noindent \textbf{Interactive Theorem Proving.} Interactive theorem provers (ITPs) offer strong correctness guarantees at the cost of substantial manual effort. Cairo and Jolt have been formalized in Lean and ACL2~\cite{avigad2022verified,kwan2024verifying}. OSS tools such as Clean~\cite{clean-github} and zkLean~\cite{zklean-github} pursue the complementary approach of co-designing circuits alongside their correctness proofs. Despite their strong correctness guarantees, they require expert-written proofs for every new component or architectural change, and scale poorly as zkVM designs evolve.

\smallskip \noindent \textbf{Relation to abstract interpretation.} \hideaki{\sys's interval evaluator can be viewed as an abstract interpretation of a finite constraint system, and its branch-and-bound structure is related to interval constraint propagation. Classical abstract interpretation typically computes an inductive invariant or fixpoint over a program’s control-flow semantics. \sys instead analyzes a bounded, already-extracted zkVM constraint instance and uses three-valued interval evaluation as a sound pruning oracle inside a finite branch-and-bound search. Its output is a cardinality classification of canonical satisfying trace tables, rather than a reachability invariant. This setting additionally requires sound treatment of finite-field wrap-around, selector-gated polynomial constraints, lookup and public-value constraints, and canonicalization of padding, permutation, and auxiliary-witness redundancy. Thus, our contribution is not interval abstraction itself, but its finite-field specialization and integration with canonicalized cardinality reasoning for zkVM constraints.}
\section{Limitations and Future Work}

\noindent
{\bf Global verification.} \sys's correctness guarantees are table-localized: given a set of consistent zkVM tables, \sys certifies the target table. Bugs from cross-table interactions, such as lookup arguments whose two sides are individually satisfied but whose interface is mis-specified, \hideaki{fall outside this guarantee. No existing automated technique verifies these cross-table interfaces. A fixed-point iteration over local certificates produced by \sys can formally compose toward global correctness, analogous to layerwise bound propagation in neural network verification~\cite{alphabetacrown} albeit at greater computational cost. Reducing this cost via smarter branch-and-bound is a natural direction for future work.}

\noindent 
{\bf Precision loss from interval abstraction.} \hideaki{Interval abstraction is sound but incomplete: over-approximation may merge feasible and infeasible assignments and therefore return \texttt{Uncertain} even when the concrete constraint system has a definite answer. This imprecision cannot produce a false certificate or a spurious bug report, because \sys prunes a region only when it is definitely infeasible and reports an under-constrained bug only after constructing two distinct concrete satisfying traces. However, it may force additional splitting and, under a finite time budget, reduce verification coverage or delay the discovery of a second witness in under-constrained cases. Tighter abstract domains or a hybrid interval–SMT procedure could address these residual cases.}

\noindent
{\bf Verifying Canonicalizer Implementation.} While the canonicalizer algorithm is formally verified in Lean 4, verifying the equivalence between the Lean specification and the actual Rust implementation is left for future work.

\noindent
{\bf Verifying non-Instruction Table.} We exclude the non-instruction table from exhaustive verification in our current prototype due to its high complexity; scaling \sys to cover it through advanced heuristics such as opcode-based slicing is an interesting direction for future work.

\section{Conclusion}

We presented \sys, a fully automated framework for testing and verifying zkVM constraint systems. We introduced a formalization of zkVM vulnerabilities that separates representational redundancy from true under- and over-constrained behaviors, and developed an interval-based branch-and-bound engine for reasoning about sets of executions.

Across five real-world zkVMs, \sys uncovers 11 previously unknown bugs, including vulnerabilities enabling control-flow hijacking and forged proofs, all missed by a state-of-the-art fuzzer. Compared to SMT-based verification, \sys achieves higher coverage (+16.5 pp) and a $51.5\times$ speedup, while \rangeverification provides up to $63\times$ efficiency gains over point-wise testing.

These results suggest that reasoning about solution-set cardinality via interval-based branch-and-bound provides a practical and scalable foundation for zkVM constraint verification.

\section*{Acknowledgments} \label{s:ack}

We sincerely thank the anonymous reviewers for their insightful feedback. We would also like to thank Tim Roughgarden and Mahimna Kelkar at The Columbia-Ethereum Research Center for Blockchain Protocol Design for constructive and helpful discussions. Hideaki is supported by Funai Fellowship and Columbia-Ethereum Center Fellowship. Junfeng is also supported by DARPA ScAN, NSF CNS 2526620, Google, Samsung, Amazon, and Columbia CAIFI (with Capital One), Enterprise AI (with Infosys), and CDFT centers.

\bibliographystyle{ACM-Reference-Format}
\bibliography{ref}

@String{Computing = "Computing" }

@String{Computer = "{IEEE} Computer" }

@String{Springer = "Springer-Verlag" }

@ArtifactSoftware{R,
    title = {R: A Language and Environment for Statistical Computing},
    author = {{R Core Team}},
    organization = {R Foundation for Statistical Computing},
    address = {Vienna, Austria},
    year = {2019},
    url = {https://www.R-project.org/},
}

@article{lavin2024survey,
  title={A Survey on the Applications of Zero-Knowledge Proofs},
  author={Lavin, Ryan and Liu, Xuekai and Mohanty, Hardhik and Norman, Logan and Zaarour, Giovanni and Krishnamachari, Bhaskar},
  journal={arXiv preprint arXiv:2408.00243},
  year={2024}
}

@inproceedings{wen2024practical,
  title={Practical Security Analysis of $\{$Zero-Knowledge$\}$ Proof Circuits},
  author={Wen, Hongbo and Stephens, Jon and Chen, Yanju and Ferles, Kostas and Pailoor, Shankara and Charbonnet, Kyle and Dillig, Isil and Feng, Yu},
  booktitle={33rd USENIX Security Symposium (USENIX Security 24)},
  pages={1471--1487},
  year={2024}
}

@article{pailoor2023automated,
  title={Automated detection of under-constrained circuits in zero-knowledge proofs},
  author={Pailoor, Shankara and Chen, Yanju and Wang, Franklyn and Rodr{\'\i}guez, Clara and Van Geffen, Jacob and Morton, Jason and Chu, Michael and Gu, Brian and Feng, Yu and Dillig, I{\c{s}}{\i}l},
  journal={Proceedings of the ACM on Programming Languages},
  volume={7},
  number={PLDI},
  pages={1510--1532},
  year={2023},
  publisher={ACM New York, NY, USA}
}

@article{belles2022circom,
  title={Circom: A circuit description language for building zero-knowledge applications},
  author={Bell{\'e}s-Mu{\~n}oz, Marta and Isabel, Miguel and Mu{\~n}oz-Tapia, Jose Luis and Rubio, Albert and Baylina, Jordi},
  journal={IEEE Transactions on Dependable and Secure Computing},
  volume={20},
  number={6},
  pages={4733--4751},
  year={2022},
  publisher={IEEE}
}

@misc{aztec2024noir,
  title = {Noir: The Universal Language of Zero-Knowledge},
  author = {{Aztec Network}},
  year = {2024},
  howpublished = {\url{https://noir-lang.org}},
  note = {Accessed: February 23, 2025}
}

@article{takahashi2025zkfuzz,
  title={zkFuzz: Foundation and Framework for Effective Fuzzing of Zero-Knowledge Circuits},
  author={Takahashi, Hideaki and Kim, Jihwan and Jana, Suman and Yang, Junfeng},
  journal={arXiv preprint arXiv:2504.11961},
  year={2025}
}

@misc{RiscZero2025SecurityDisclosure,
  author       = {{RISC Zero}},
  title        = {Security Disclosure: A missing constraint was recently discovered in the rv32im circuit},
  howpublished = {\url{https://x.com/RiscZero/status/1935404812146725042}},
  month        = jun,
  year         = 2025,
  note         = {RISC-V instructions in risc0-zkvm versions 2.0.0–2.0.2 were affected; patch released in version 2.1.0} 
}

@article{kwan2024verifying,
  title={Verifying jolt zkVM lookup semantics},
  author={Kwan, Carl and Dao, Quang and Thaler, Justin},
  journal={Cryptology ePrint Archive},
  year={2024}
}

@article{dokchitser2023zero,
  title={Zero knowledge virtual machine step by step},
  author={Dokchitser, Tim and Bulkin, Alexandr},
  journal={Cryptology ePrint Archive},
  year={2023}
}

@article{thomas2025valida,
  title={Valida ISA Spec, version 1.0: A zk-Optimized Instruction Set Architecture},
  author={Thomas, Morgan and Ratsimbazafy, Mamy and Bugaj, Marcin and Revill, Lewis and Modica, Carlo and Schmidt, Sebastian and Tan, Ventali and Lubarov, Daniel and Gillett, Max and Dai, Wei},
  journal={arXiv preprint arXiv:2505.08114},
  year={2025}
}

@article{hassanzadeh2025constraint,
  title={Constraint-Level Design of zkEVMs: Architectures, Trade-offs, and Evolution},
  author={Hassanzadeh-Nazarabadi, Yahya and Taheri-Boshrooyeh, Sanaz},
  journal={arXiv preprint arXiv:2510.05376},
  year={2025}
}

@article{habock2022multivariate,
  title={Multivariate lookups based on logarithmic derivatives},
  author={Hab{\"o}ck, Ulrich},
  journal={Cryptology ePrint Archive},
  year={2022}
}

@techreport{openvm2025whitepaper,
  title        = {OpenVM Whitepaper},
  author       = {OpenVM Contributors},
  institution  = {OpenVM},
  type         = {Whitepaper},
  month        = mar,
  day          = 31,
  year         = {2025},
  url          = {https://openvm.dev/whitepaper.pdf},
}

@article{fernandez2006interval,
  title={An interval constraint branching scheme for lattice domains.},
  author={Fern{\'a}ndez, Antonio J and Hill, Patricia M},
  journal={J. Univers. Comput. Sci.},
  volume={12},
  number={11},
  pages={1466--1499},
  year={2006}
}

@article{benno2026jolt,
  title={Jolt Atlas: Verifiable Inference via Lookup Arguments in Zero Knowledge},
  author={Benno, Wyatt and Centelles, Alberto and Douchet, Antoine and Gibran, Khalil},
  journal={arXiv preprint arXiv:2602.17452},
  year={2026}
}

@article{wang2026zkagent,
  title={zkAgent: Verifiable Agent Execution via One-Shot Complete LLM Inference Proof},
  author={Wang, Lizheng and Lou, Hancheng and Li, Chongrong and Yu, Yu and Hu, Yuncong},
  journal={Cryptology ePrint Archive},
  year={2026}
}

@article{ron2026verifiable,
  title={Verifiable Provenance of Software Artifacts with Zero-Knowledge Compilation},
  author={Ron, Javier and Monperrus, Martin},
  journal={arXiv preprint arXiv:2602.11887},
  year={2026}
}

@article{ke2025consistency,
  title={Consistency Verification for Zero-Knowledge Virtual Machine on Circuit-Irrelevant Representation},
  author={Ke, Jingyu and Liang, Boxuan and Li, Guoqiang},
  journal={Cryptology ePrint Archive},
  year={2025}
}

@article{hochrainer2025arguzz,
  title={Arguzz: Testing zkvms for soundness and completeness bugs},
  author={Hochrainer, Christoph and W{\"u}stholz, Valentin and Christakis, Maria},
  journal={arXiv preprint arXiv:2509.10819},
  year={2025}
}

@article{goldberg2021cairo,
  title={Cairo--a Turing-complete STARK-friendly CPU architecture},
  author={Goldberg, Lior and Papini, Shahar and Riabzev, Michael},
  journal={Cryptology ePrint Archive},
  year={2021}
}

@article{ben2018scalable,
  title={Scalable, transparent, and post-quantum secure computational integrity},
  author={Ben-Sasson, Eli and Bentov, Iddo and Horesh, Yinon and Riabzev, Michael},
  journal={Cryptology ePrint Archive},
  year={2018}
}

@misc{ziren2025,
  title        = {Ziren},
  author       = {{ProjectZKM}},
  howpublished = {\url{https://github.com/ProjectZKM/Ziren}},
  note         = {Accessed: 2026-04-21}
}

@misc{sphinx2025,
  title        = {Sphinx},
  author       = {{Argument Computer}},
  howpublished = {\url{https://github.com/argumentcomputer/sphinx}},
  note         = {Accessed: 2026-04-21}
}

@misc{pico2025,
  title        = {Pico},
  author       = {{Brevis Network}},
  howpublished = {\url{https://github.com/brevis-network/pico}},
  note         = {Accessed: 2026-04-21}
}

@misc{sp12026,
  title        = {SP1},
  author       = {{Succinct Labs}},
  year         = {2026},
  howpublished = {\url{https://github.com/succinctlabs/sp1}},
  note         = {Accessed: 2026-04-21}
}

@misc{valida2025,
  title        = {Valida VM},
  author       = {{Lita Foundation}},
  howpublished = {\url{https://github.com/lita-xyz/valida-vm}},
  note         = {Accessed: 2026-04-21}
}

@misc{valida2025core,
  title        = {Valida (Old Version)},
  author       = {{Valida}},
  howpublished = {\url{https://github.com/valida-xyz/valida}},
  note         = {Accessed: 2026-04-21}
}

@misc{veridise-sp1-picus,
  title={Verifying {SP1} circuit determinism with {Picus}: A collaboration between {Veridise} and {Succinct}},
  author={Pailoor, Shankara},
  year={2025},
  howpublished={\url{https://veridise.com/blog/audit-insights/verifying-sp1-circuit-determinism-with-picus-a-collaboration-between-veridise-and-succinct/}},
  note={Accessed 2026-04-23}
}

@misc{audithub-ziren-picus,
  title={Verifying {Ziren}'s circuits with {Picus}: Another step towards deterministic {zkVMs}},
  author={Pailoor, Shankara},
  year={2025},
  howpublished={\url{https://audithub.dev/blog/picus/verifying-zirens-circuits-with-picus-another-step-towards-deterministic-zkvms/}},
  note={Accessed 2026-04-23}
}

@misc{clean-github,
  title={{Clean}: A verified zk{EVM}},
  author={{Verified-zkEVM}},
  howpublished={\url{https://github.com/Verified-zkEVM/clean/}},
  note={Accessed 2026-04-23}
}

@misc{zklean-github,
  title={zk{Lean}},
  author={{Galois, Inc.}},
  howpublished={\url{https://github.com/GaloisInc/zkLean}},
  note={Accessed 2026-04-23}
}

@inproceedings{avigad2022verified,
  title={A verified algebraic representation of cairo program execution},
  author={Avigad, Jeremy and Goldberg, Lior and Levit, David and Seginer, Yoav and Titelman, Alon},
  booktitle={Proceedings of the 11th ACM SIGPLAN International Conference on Certified Programs and Proofs},
  pages={153--165},
  year={2022}
}

@inproceedings{alphabetacrown,
  title     = {Neural Network Verification with Branch-and-Bound 
               for General Nonlinearities},
  author    = {Shi, Zhouxing and Jin, Qirui and Kolter, J. Zico 
               and Jana, Suman and Hsieh, Cho-Jui and Zhang, Huan},
  booktitle = {Tools and Algorithms for the Construction and 
               Analysis of Systems (TACAS)},
  year      = {2025},
  publisher = {Springer}
}

\section{Ethical Considerations} We responsibly disclosed all previously unknown vulnerabilities discovered in public projects. We first contacted the maintainers privately and created public GitHub issues only with their approval. We also worked closely with each team to assess the severity of the issues and support the development of appropriate fixes.

%ADD
%we help devs fix problems.....

%% Appendices
\ifarxiv
  \appendix
  \section{Overview of zkVMs}
\label{appendix:sp}

\paragraph{Sparsity in zkVMs}

\begin{table}[h]
\centering
\caption{Constraint density $\rho = T_{\mathrm{arith}}(\mathcal{C}) / (|\mathcal{C}| \cdot m)$
measured across five real-world zkVMs. $|\mathcal{C}|$ is the number of AIR
constraints, $m$ is the total number of trace cell variables, and $T_{\mathrm{arith}}$ is the arithmetic gate count after common-subexpression elimination. All values are means across the evaluated tables per zkVM. Lower $\bar{\rho}$ indicates sparser constraint graphs.}
\label{tab:density}
\begin{tabular}{lrrrrr}
\toprule
zkVM & \#tables & $\overline{m}$ & $\overline{|\mathcal{C}|}$ & $\overline{T_{\mathrm{arith}}}$ & $\bar{\rho}$ \\
\midrule
Sphinx & 8  & 63.1 & 76.9  & 242.0  & 0.081  \\
Pico   & 9  & 59.3 & 79.1  & 237.3  & 0.086  \\
Ziren  & 13 & 41.8 & 51.7  & 171.5  & 0.107  \\
Valida & 9  & 36.2 & 33.8  & 139.2  & 0.187  \\
SP1    & 11 & 42.1 & 55.7  & 501.4  & 0.227  \\
\midrule
\textbf{Overall} & \textbf{50} & \textbf{47.4} & \textbf{58.3} & \textbf{261.4} & \textbf{0.140} \\
\bottomrule
\end{tabular}
\end{table}

Across 50 tables from five production zkVMs (Table~\ref{tab:density}), we observe $\rho$ ranging from $0.017$ (Sphinx CPU) to $0.809$ (SP1 Branch), with an overall mean of $\bar{\rho} = 0.140$ --- the typical constraint uses only {14}\% of its theoretical connectivity. Three of the five zkVMs (Sphinx, Pico, Ziren) are consistently sparse with $\bar{\rho} \leq 0.11$; SP1 and Valida contain a small number of dense tables (notably SP1 Branch and Valida Mul32) that raise their per-zkVM averages, but remain below $\bar{\rho} = 0.23$ overall.

\paragraph{Example Arithmetization}

Fig.~\ref{fig:examples-constraints} shows the example arithmetization of a zkVM.

\begin{figure*}[!th]
    \centering
    \includegraphics[width=0.95\linewidth]{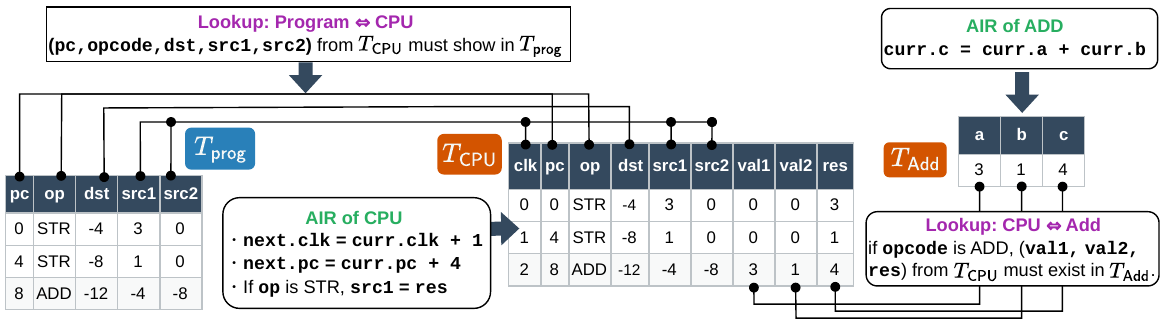}
    \caption{Conceptual arithmetization of a zkVM execution. The Program Table ($T_{\mathsf{prog}}$) is an instance table defining the compiled instructions. The CPU Table ($T_{\mathsf{CPU}}$) records the cycle-by-cycle execution state, and its AIR enforces the main CPU logic, such as valid transitions of \texttt{pc} and \texttt{clk}. The ADD Table ($T_\mathsf{Add}$) validates the correctness of the addition operation. Global consistency is enforced by lookup constraints, ensuring all entries in $T_{\mathsf{CPU}}$ match auxiliary tables.}
    \label{fig:examples-constraints}
\end{figure*}

\section{Design Details}

\begin{figure}[!th]
    \centering
    \includegraphics[width=0.99\linewidth]{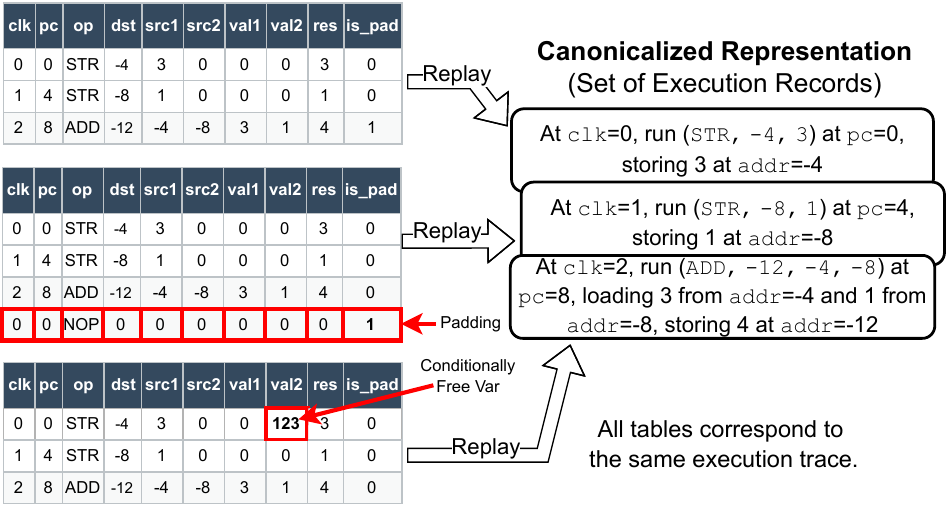}
    \caption{To resolve representational redundancies in trace tables like null-row padding and conditionally free variables, a canonicalizer ($\mathcal{R}$) maps trace tables to a unique set of semantic execution records, ensuring a one-to-one mapping.}
    \label{fig:placeholder}
\end{figure}

\subsection{Sound Over-Approximation of Constraints}
\label{appendix:lookup}

This section provides the rigorous formalization of the interval-based sound over-approximations used by \sys. We transition from concrete constraints over the finite field $\mathbb{F}_p$ to an abstract interval domain $\mathbb{I}$, ensuring that the conditions defined in Definition~\ref{def-soa} are strictly maintained for all lookup categories.

\subsection{Preliminaries and Interval Arithmetic}
\label{appendix:subsec:ia}

We first provide several additional notations related to intervals.

\textbf{1) Interval Membership}: For an integer $a \in \mathbb{Z}$ and an interval $[L, R] \in \mathbb{I}$, we denote $a \in [L, R]$ if $L \leq a \leq R$, and $a = [L, R]$ if $a = L = R$. 

\textbf{2) Interval Arithmetic}: The basic arithmetic operations, add, sub, and mul, are defined as: \begin{align*}
    [a, b] + [c, d] &= [a + c, b + d] \\
    [a, b] - [c, d] &= [a - d, b - c] \\
    [a, b] \cdot[c, d] &= [\min(ac, ad, bc, bd), \max(ac, ad, bc, bd)]
\end{align*}

\textbf{3) Row Membership}: We view a table as a set of rows. For instance, when a row vector $\vec{v} \in \mathbb{F}^{W}$ appears as a row of a table $T \in \mathbb{F}_p^{H \times W}$, we write $\vec{v} \in T$.

\textbf{4) Interval Area}: Let $|[L, R]| = R - L + 1$, meaning the area of the interval $[L, R]$. For a vector of intervals $\vec{I} = (I_1, I_2, \dots, I_w)$, we define $|\vec{I}| = \prod_{j=1}^{w} |I_j|$.

\textbf{5) Column Projection} We define a column projection $\mathfrak{f}_i$ as a functional that extracts a value from the table based on row $i$. Formally, $\mathfrak{f}_i(T)$ returns a linear combination of cells in the $i$-th row: \begin{equation}
    \mathfrak{f}_i(T) = \sum_{j=0}^{W-1} \alpha_j \cdot T_{i,j}
\end{equation}, where $\alpha_j \in \mathbb{F}_p$ is a fixed coefficient specific to each $\mathfrak{f}$. In the interval domain, $\mathfrak{f}_i(\widetilde{T})$ is similarly computed using interval arithmetic.

\textbf{6) Vector Projections}: We define a vector projection $\vec{\mathfrak{f}}_i$ as a tuple of $k$ column projections:

\begin{equation}
    \vec{\mathfrak{f}}_i(T) = \left( \mathfrak{f}_{i,1}(T), \mathfrak{f}_{i,2}(T), \dots, \mathfrak{f}_{i,k}(T) \right) \in \mathbb{F}_p^k
\end{equation} In the interval domain, $\vec{\mathfrak{f}}_i(\widetilde{T}) = \left( I_{i,1}, I_{i,2}, \dots, I_{i,k} \right) \in \mathbb{I}^k$ is the Cartesian product of the resulting intervals.

\textbf{7) Interval Disjointness} Two intervals $I, J$ are disjoint ($I \cap J = \emptyset$) if $[[I]]_p \cap [[J]]_p = \emptyset$. An interval vector $\vec{I} \in \mathbb{F}_p^{W}$ is disjoint from a table $T \subseteq \mathbb{F}_p^{H \times W}$, written as $[[\vec{I}]]_p \cap T = \emptyset$, if no vector $v \in [[\vec{I}]]_p$ exists in $T$.

\textbf{8) Vector Embedding} For a vector of intervals $\vec{I} = (I_1, \dots, I_k)$, we also define a vector embedding $[[\vec{I}]]_p = [[I_1]]_p \times \dots \times [[I_k]]_p \subseteq \mathbb{F}_p^k$.

\subsection{Sound Over-Approximation of Lookup Constraints}

We denote $\mathfrak{s}_i(\tilde{T})$ as the interval projection of the selector for the $i$-th row. A row is considered "active" if $1 \in \mathfrak{s}_i(\tilde{T})$.

\paragraph{L-1) Program Lookup}: The Program Lookup $\mathcal{L}_{\mathrm{prog}}$ ensures that the executed instruction is consistent with the static program binary $T_{\text{prog}}$. Let $\vec{prog}_i$ be the vector projection of the columns corresponding to the program counter, opcode, and operands. \begin{equation}
\mathcal{L}_{\mathrm{prog}}(T) := \forall{i}.\;\, \isreal_i(T) = 1 \implies \overrightarrow{\mathfrak{prog}}_i(T) \in T_{\mathrm{prog}}
\end{equation}

The sound over-approximation $\widetilde{\mathcal{L}_{\mathrm{prog}}}$ returns $\mathsf{True}$ if any active row's program-related columns collapse into a singleton vector that is present in $T_{\mathrm{prog}}$, and returns $\mathsf{False}$ if the interval vector of program-related columns is entirely disjoint from any row in $T_{\mathrm{prog}}$. In other cases, $\widetilde{\mathcal{L}_{\mathrm{prog}}}$ returns $\mathsf{Uncertain}$.

\begin{align}
\widetilde{\mathcal{L}}_{\mathrm{prog}}(\widetilde{T}) := \begin{cases}
\mathsf{True}& \text{if}\;\; \forall{i}.\; 1 \in \isreal_i(\widetilde{T}) \implies \\ &\quad\;|\overrightarrow{\mathfrak{prog}}_i(\widetilde{T})| = 1 \land \\ &\quad\;\overrightarrow{\mathfrak{prog}}_i(\widetilde{T}) \in T_{\mathrm{prog}} \\
\mathsf{False}& \text{if } \exists{i}.\; [1, 1] = \isreal_i(\widetilde{T}) \text{ and } \\& \quad\; \llbracket \overrightarrow{\mathfrak{prog}}_i(\widetilde{T}) \rrbracket_p \cap T_{\mathrm{prog}} = \emptyset \\
\mathsf{Uncertain}& \text{otherwise}
\end{cases}
\end{align}

\paragraph{L-2) Byte Lookup}: The Byte lookups enforce precomputed relations over a restricted range $[L, R]$, typically $[0, 255]$. Let $\mathfrak{a}_i, \mathfrak{b}_i, \mathfrak{c}_i$ be the projections for the output and operands of a byte operation $\star$. Then the semantic meaning of the Byte Lookup constraints can be written as:

\begin{align}
    \mathcal{L}_{\mathrm{byte}}(T) &:= \forall{i}.\;\isreal_i(T) = 1 \implies \\& 
    L \leq \opb_i(T) \leq R \;\,\land\;\, L \leq \opc_i(T) \leq R \;\,\land \;\, \\ &\label{eq:abc} \opa_i(T) = \opb_i(T) \,\star\, \opc_i(T)
\end{align} \noindent The Range-check operation is treated as a special case where Eq.~\ref{eq:abc} is omitted, effectively enforcing only the range constraint.

The sound over-approximation $\widetilde{\mathcal{L}_{\mathrm{byte}}}$ semantically lifts each byte operation $\star$ into an over-approximated operator $\tilde{\star}$ over intervals that uses bit-masking to track fixed and varying bits to derive tight bounds (see Tab.~\ref{tab:star}). $\widetilde{\mathcal{L}_{\mathrm{byte}}}$ returns $\mathsf{True}$ if the operand intervals fit entirely within the permitted byte range and the output interval matches the result of the lifted operation as a singleton. Conversely, it returns $\mathsf{False}$ if any operand is disjoint from the required range or if the claimed result interval is disjoint from the possible outcomes of the interval-based approximation of bitwise operation $\tilde{\star}$. In other cases, it just returns $\mathsf{Uncertain}$.

\begin{align}
    \widetilde{\mathcal{L}}_{\mathrm{byte}}&(T) := \begin{cases}
        \mathsf{True}  & \text{if }\; \forall{i.}\;1 \in \isreal_i(\widetilde{T}) \implies \\ &\;\;\opb_i \subseteq [L, R] \land \opc_i \subseteq [L, R] \land \\ &\;\;|\opa_i(\widetilde{T})| = 1 \land \opa_i(\widetilde{T}) = \opb_i(\widetilde{T}) \,\widetilde{\star}\, \opc_i(\widetilde{T})  \\
        \mathsf{False} & \text{if }\; \exists{i.}\; [1, 1] = \isreal_i(\widetilde{T}) \;\land\; \\ &\;\;\Big( \begin{aligned}
& \opb_i \not\subseteq [L, R]
\;\lor\; \opc_i \not\subseteq [L, R] \;\lor\\
& \opa_i(\widetilde{T}) \cap
(\opb_i(\widetilde{T}) \,\widetilde{\star}\, \opc_i(\widetilde{T}))
= \emptyset
\end{aligned} \Big) \\ 
        \mathsf{Uncertain} & \text{otherwise }\;
    \end{cases}
\end{align}

\begin{table*}[!th]
\centering
\caption{Sound Over-Approximated Constraints for Byte Lookup Tables. Here $a_{lo}, a_{hi}$ (resp.\ $b_{lo}, b_{hi}$) are the lower and upper bounds of the first (resp.\ second) byte operand interval; $v = (a_{lo} \oplus a_{hi}) \mid (b_{lo} \oplus b_{hi})$ is the varying-bit mask; all values are in $[0, 255]$.}
\label{tab:star}
\begin{tabular}{@{}lll@{}}
\toprule
Opcode & Interval Approximation ($\tilde{\star}$) & Soundness Notes \\ \midrule
\texttt{BitAnd} & $[(a_{lo} \mathbin{\&} b_{lo}) \mathbin{\&} \neg v,\;\; \min(a_{hi}, b_{hi})]$ & \begin{tabular}[c]{@{}l@{}}Bits forced to 0 in either operand cannot appear in result;\\ upper bound uses $\min(a_{hi}, b_{hi})$ since AND cannot exceed \\ either input.\end{tabular} \\ \hline
\texttt{BitOr} & $[(a_{lo} \mid b_{lo}),\;\; (a_{hi} \mid b_{hi}) \mid v]$ & \begin{tabular}[c]{@{}l@{}}Lower bound: bits set in both $a_{lo}$ and $b_{lo}$ are definitely set;\\ upper bound: all bits present in either operand's high end, \\ extended by varying bits.\end{tabular} \\ \hline
\texttt{BitXor} & $[(a_{lo} \oplus b_{lo}) \mathbin{\&} \neg v,\;\; (a_{lo} \oplus b_{lo}) \mid v]$ & \begin{tabular}[c]{@{}l@{}}Fixed bits (those not in $v$) produce a known XOR result;\\ varying bits are treated conservatively as either 0 or 1.\end{tabular} \\ \hline
\texttt{Lt} & \begin{tabular}[c]{@{}l@{}} $[1, 1]$ if $a_{hi} < b_{lo}$;\quad $[0, 0]$ if $a_{lo} \ge b_{hi}$;\quad else $[0, 1]$ \end{tabular} & \begin{tabular}[c]{@{}l@{}}Exact when intervals are non-overlapping; uncertain\\ when the comparison outcome depends on the concrete values.\end{tabular} \\ \hline
\texttt{MSB} & $[1,\; 1]$ if $a_{lo} \ge 128$;\quad $[0,\; 0]$ if $a_{hi} < 128$;\quad else $[0, 1]$ & \begin{tabular}[c]{@{}l@{}}Evaluates bit 7 (sign bit) of an 8-bit byte;\\ exact when the interval lies entirely above or below 128.\end{tabular} \\
\bottomrule
\end{tabular}
\end{table*}

\paragraph{L-3) ALU Lookup}: For lookups involving the Arithmetic Logic Unit (ALU), we reduce the constraints to semantically equivalent algebraic equations over integer intervals. Specifically, let $\mathfrak{a}_i$, $\mathfrak{b}_i$, and $\mathfrak{c}_i$ be the column projections corresponding to the output, the first input operand, and the second input operand, respectively. The concrete constraint semantically ensures that the output is the result of the VM-level operation $\diamond \in \{+, -, \times, \div, \dots\}$ applied to the inputs:

\begin{equation}
\mathcal{L}_{\mathrm{alu}}(T) := \forall{i}.\; \mathfrak{s}_i(T) = 1 \implies \mathfrak{a}_i(T) = \mathfrak{b}_i(T) ,\diamond, \mathfrak{c}_i(T)
\end{equation}

The sound approximation $\widetilde{\mathcal{L}}_{\mathrm{alu}}(\widetilde{T})$ leverages a library of sound interval approximations $\widetilde{\diamond}$ for each ALU operation (see Tab.~\ref{tab:diamond}). These transformers map the input interval to the smallest possible interval bound. The approximation is formally defined as:

\begin{equation}
\widetilde{\mathcal{L}}_{\mathrm{alu}}(\widetilde{T}) := \begin{cases}
    \mathsf{True} &\text{if }\,\forall{i}.\;\, 1 \in \isreal_i(\widetilde{T}) \implies \\ &\;\;\;\;\,|\opa_i(\widetilde{T})| = 1 \land \opa_i(\widetilde{T}) = \opb_i(\widetilde{T}) \,\widetilde{\diamond}\, \opc_i(\widetilde{T}) \\
    \mathsf{False} &\text{if }\,\exists{i}.\;\, [1, 1] = \isreal_i(\widetilde{T}) \text{ and } \\ &\;\;\;\;\,\opa_i(\widetilde{T}) \cap (\opb_i(\widetilde{T}) \,\widetilde{\diamond}\, \opc_i(\widetilde{T})) = \emptyset \\
    \mathsf{Uncertain}
\end{cases}
\end{equation}

\noindent In the $\mathsf{False}$ case, the evaluator identifies that the interval of the claimed result $\mathfrak{a}_i(\widetilde{T})$ is disjoint from the range of possible outcomes produced by the operands under the operation $\widetilde{\diamond}$.

% \usepackage{booktabs}
\begin{table*}[!th]
\centering
\caption{Sound Over-Approximated Constraints for ALU Lookup Tables. $W = 2^{32}$; $a_{lo}, a_{hi}$ and $b_{lo}, b_{hi}$ are the lower and upper bounds of the two 32-bit input operand intervals. Signed operations first reinterpret the unsigned interval as a two 's-complement signed interval via the mapping $x \mapsto x - W$ for $x \ge 2^{31}$. ``Full word'' denotes $[0, W-1]$.}
\label{tab:diamond}
\begin{tabular}{@{}lll@{}}
\toprule
Opcode & Interval Approximation ($\tilde{\diamond}$) & Soundness Notes \\ \midrule
\texttt{AddU} & \begin{tabular}[c]{@{}l@{}} $[a_{lo}{+}b_{lo},\; a_{hi}{+}b_{hi}]$ if no overflow;\\ $[a_{lo}{+}b_{lo}{-}W,\; a_{hi}{+}b_{hi}{-}W]$ if overflow certain;\\ full word otherwise \end{tabular} & \begin{tabular}[c]{@{}l@{}}Exact for singletons (mod $W$); overflow is certain when\\ $a_{lo}+b_{lo} \ge W$, impossible when $a_{hi}+b_{hi} < W$.\end{tabular} \\ \hline
\texttt{SubU} & \begin{tabular}[c]{@{}l@{}} $[a_{lo}{-}b_{hi},\; a_{hi}{-}b_{lo}]$ if no underflow;\\ $[a_{lo}{+}W{-}b_{hi},\; a_{hi}{+}W{-}b_{lo}]$ if underflow certain;\\ full word otherwise \end{tabular} & \begin{tabular}[c]{@{}l@{}}Exact for singletons (mod $W$); underflow is certain when\\ $a_{hi} < b_{lo}$, impossible when $a_{lo} \ge b_{hi}$.\end{tabular} \\ \hline
\texttt{Mul} & $[a_{lo} \times b_{lo},\; a_{hi} \times b_{hi}] \bmod W$ & \begin{tabular}[c]{@{}l@{}}Exact for singletons (mod $W$); returns full word when\\ $a_{hi} \times b_{hi} \ge W$ and operands are not singletons.\end{tabular} \\ \hline
\texttt{MulH\,/\,MulHS} & $[\min(\mathcal{C}) \gg 32,\;\; \max(\mathcal{C}) \gg 32] \bmod W$ & \begin{tabular}[c]{@{}l@{}}Signed high-word: $\mathcal{C} = \{a_{lo}b_{lo},\,a_{lo}b_{hi},\,a_{hi}b_{lo},\,a_{hi}b_{hi}\}$\\ over signed 32-bit interpretation; arithmetic right-shift \\ by 32.\end{tabular} \\ \hline
\texttt{MulHU} & $[\lfloor a_{lo} \times b_{lo} / W \rfloor,\;\; \lfloor a_{hi} \times b_{hi} / W \rfloor]$ & \begin{tabular}[c]{@{}l@{}}Unsigned high-word: logical right-shift of 64-bit \\ product by 32; monotone in both operands so \\ endpoint products bound the result.\end{tabular} \\ \hline
\texttt{MulTL\,/\,MulTH} & \begin{tabular}[c]{@{}l@{}} $\text{lo}(P) \in [\min(\mathcal{C}) \mathbin{\&} (W{-}1),\; \max(\mathcal{C}) \mathbin{\&} (W{-}1)]$\\ $\text{hi}(P) \in [\min(\mathcal{C}) \gg 32,\; \max(\mathcal{C}) \gg 32]$ \end{tabular} & \begin{tabular}[c]{@{}l@{}}Signed full 64-bit product split into lower (TL) and \\ upper (TH) 32-bit halves; \\ $\mathcal{C}$ over signed interpretation of both intervals.\end{tabular} \\ \hline
\texttt{MulTUL\,/\,MulTUH} & \begin{tabular}[c]{@{}l@{}} $\text{lo}(P) \in [a_{lo}b_{lo} \mathbin{\&} (W{-}1),\; a_{hi}b_{hi} \mathbin{\&} (W{-}1)]$\\ $\text{hi}(P) \in [a_{lo}b_{lo} \gg 32,\; a_{hi}b_{hi} \gg 32]$ \end{tabular} & \begin{tabular}[c]{@{}l@{}}Unsigned full 64-bit product split into lower (TUL) and \\ upper (TUH) 32-bit halves; monotone so endpoint \\ products suffice.\end{tabular} \\ \hline
\texttt{And} & $[0,\;\; \min(a_{hi}, b_{hi})]$ & \begin{tabular}[c]{@{}l@{}}Exact for singletons; lower bound is 0 (AND can clear \\ all bits); upper bound $\min(a_{hi}, b_{hi})$ since AND cannot \\ exceed either input.\end{tabular} \\ \hline
\texttt{Or} & \begin{tabular}[c]{@{}l@{}} 
$[a_{lo} | b_{lo}, \;\; a_{hi} | b_{hi}]$ $\quad\;\;\;$ if $a$ and $b$ are singletons \\ 
$[\max(a_{lo}, b_{lo}),\;\; 2^{\lceil \log_2((a_{hi} \mid b_{hi})+1)\rceil} - 1]$ $\quad$ else \end{tabular} & \begin{tabular}[c]{@{}l@{}}Exact for singletons; lower bound: bits set in the larger \\ operand are preserved; upper bound: \\ next power-of-two minus 1 covering $a_{hi}|b_{hi}$.\end{tabular} \\ \hline
\texttt{Xor} & \begin{tabular}[c]{@{}l@{}}$[a_{lo} \oplus b_{lo}, \;\; a_{hi} \oplus b_{hi}]$ $\quad\;\;\;$ if $a$ and $b$ are singletons \\ $[0,\;\; 2^{\lceil \log_2((a_{hi} \mid b_{hi})+1)\rceil} - 1]$ else \end{tabular} & \begin{tabular}[c]{@{}l@{}}Exact for singletons; XOR can produce any value up \\ to the highest bit set in either operand, \\ hence lower bound is 0.\end{tabular} \\ \hline
\texttt{Div} & $[\lfloor a_{lo} / b_{hi} \rfloor,\;\; \lfloor a_{hi} / b_{lo} \rfloor]$ & \begin{tabular}[c]{@{}l@{}}Unsigned floor division; returns full word if $b_{lo} = 0$\\ (division by zero possible).\end{tabular} \\ \hline
\texttt{SDiv} & $[\min(\mathcal{C}),\;\; \max(\mathcal{C})]$ & \begin{tabular}[c]{@{}l@{}}Signed division; $\mathcal{C} = \{a_{lo}/b_{lo},\,a_{lo}/b_{hi},\,a_{hi}/b_{lo},\,a_{hi}/b_{hi}\}$\\ over signed interpretation; full word if $0 \in [b_{lo}, b_{hi}]$\\ or if $a_{lo} = \mathtt{INT\_MIN}$ and $b = [-1,-1]$.\end{tabular} \\ \hline
\texttt{Eq\,/\,NEq} & \begin{tabular}[c]{@{}l@{}} $[1, 1]$ if $a$ and $b$ are equal singletons;\\ $[0, 0]$ if $[a_{lo}, a_{hi}] \cap [b_{lo}, b_{hi}] = \emptyset$;\\ $[0, 1]$ otherwise \end{tabular} & \begin{tabular}[c]{@{}l@{}}Eq and NEq are duals: $\widetilde{\mathsf{NEq}}$ inverts the singleton \\ cases; exact whenever the intervals are disjoint \\ or both singletons.\end{tabular} \\ \hline
\texttt{Lt} & $[1,1]$ if $a_{hi} < b_{lo}$;\quad $[0,0]$ if $a_{lo} \ge b_{hi}$;\quad else $[0, 1]$ & \begin{tabular}[c]{@{}l@{}}Unsigned less-than over 32-bit words; exact when \\ intervals are entirely on one side of each other.\end{tabular} \\ \hline
\texttt{SLt} & $[1,1]$ if $a_{hi} < b_{lo}$;\quad $[0,0]$ if $a_{lo} \ge b_{hi}$;\quad else $[0, 1]$ & \begin{tabular}[c]{@{}l@{}}Signed less-than; bounds interpreted as two's \\-complement signed 32-bit integers before comparison.\end{tabular} \\ \hline
\texttt{SLe} & $[1,1]$ if $a_{hi} \le b_{lo}$;\quad $[0,0]$ if $a_{lo} > b_{hi}$;\quad else $[0, 1]$ & \begin{tabular}[c]{@{}l@{}}Signed less-than-or-equal; same as SLt but with \\ non-strict inequality, so the certain-true condition uses \\ $\le$ instead of $<$.\end{tabular} \\ \hline
\texttt{SRL} & \begin{tabular}[c]{@{}l@{}} $[0, 0]$ if $b_{lo} \ge 32$;\\ $[a_{lo} \gg b_{hi},\;\; a_{hi} \gg b_{lo}]$ if $b_{hi} < 32$;\\ $[0,\;\; a_{hi} \gg b_{lo}]$ otherwise \end{tabular} & \begin{tabular}[c]{@{}l@{}}Logical right-shift by a possibly-uncertain amount; \\ shifting by $\ge 32$ always yields 0; shifting by a range \\ $[b_{lo}, b_{hi}]$ yields minimum at the maximum shift and \\ maximum at the minimum shift.\end{tabular} \\
\bottomrule
\end{tabular}
\end{table*}

\paragraph{L-4) Control Flow Lookup}: Let $\programcounter_i$ and $\nextprogramcounter_i$ be column projections on the $i$-th row. The semantic meaning of the Control Flow Lookup is expressed via the opcode-specific function $\mathsf{cf}$, which computes the next program counter from the current PC and operands:

\begin{align}
\mathcal{L}_{\mathsf{cf}}(T)& := \forall{i}.\; \isreal_i(T) = 1 \implies \\ &\nextprogramcounter_i(T) = \mathsf{cf}(\programcounter_i(T), \opa_i(T), \opb_i(T), \opc_i(T))
\end{align}

We naturally derive the sound interval approximation $\widetilde{\mathsf{cf}}$ based on $\widetilde{\diamond}$, and the sound over-approximation $\widetilde{\mathcal{L}_{\mathrm{cf}}}$ returns $\mathsf{True}$ when $\text{pc}'$ is a singleton such that $\text{pc}' = \widetilde{\mathrm{cf}_{\circ}}({\text{pc}, a, b, c})$, $\mathsf{False}$ when $\text{pc}'$ and $\widetilde{\mathrm{cf}_{\circ}}({\text{pc}, a, b, c})$ are entirely disjoint, and $\mathsf{Uncertain}$ in other cases.

\begin{align}
&\widetilde{\mathcal{L}}_{\mathsf{cf}}(\widetilde{T}) \\ &:=\begin{cases}
    \mathsf{True}  &\text{if }\;\forall{i}.\;1\in \isreal_i(\widetilde{T}) \implies \\ &\quad|\nextprogramcounter_i(\widetilde{T})| = 1 \land \\&\quad\nextprogramcounter_i(\widetilde{T}) = \widetilde{\mathsf{cf}}(\programcounter_i(\widetilde{T}), \opa_i(\widetilde{T}), \opb_i(\widetilde{T}), \opc_i(\widetilde{T})) \\
    \mathsf{False} &\text{if }\;\exists{i}.\; [1, 1] = \isreal_i(\widetilde{T}) \;\land \\&\quad\nextprogramcounter_i(\widetilde{T}) \cap \widetilde{\mathsf{cf}}(\programcounter_i(\widetilde{T}), \opa_i(\widetilde{T}), \opb_i(\widetilde{T}), \opc_i(\widetilde{T})) = \emptyset \\
    \mathsf{Uncertain} &\text{otherwise} \\
\end{cases}
\end{align}

Table~\ref{tab:cf} lists the $\widetilde{\mathsf{cf}}$ approximation for each implemented control-flow opcode. In all branch instructions, $\opa$ and $\opb$ are the two operands compared by the branch condition, $\opc$ is the signed branch offset (relative to the current PC), and $\delta$ is the default sequential step size (e.g., 4 for RISC-V). The condition evaluation reuses the interval predicates $\widetilde{\mathsf{Eq}}$, $\widetilde{\mathsf{SLt}}$, $\widetilde{\mathsf{SLe}}$, and $\widetilde{\mathsf{Lt}}$ from Table~\ref{tab:diamond}.

\begin{table*}[!th]
\centering
\caption{Sound Over-Approximated Next-PC Computation $\widetilde{\mathsf{cf}}$ for Control-Flow Opcodes. $\mathit{pc}$ denotes the current program counter interval; $\delta$ is the default step size; $\oplus$ denotes unsigned modular addition (\texttt{AddU}). Branch outcomes are computed using the three-valued predicates from Table~\ref{tab:diamond}.}
\label{tab:cf}
\begin{tabular}{@{}lll@{}}
\toprule
Opcode & $\widetilde{\mathsf{cf}}(\mathit{pc}, \opa, \opb, \opc)$ & Notes \\ \midrule
\texttt{BEQ} & \begin{tabular}[c]{@{}l@{}} $\mathit{pc} \oplus \opc$ \quad if $\widetilde{\mathsf{Eq}}(\opa, \opb) = [1,1]$;\\ $\mathit{pc} \oplus [\delta, \delta]$ \quad if $\widetilde{\mathsf{Eq}}(\opa, \opb) = [0,0]$;\\ $(\mathit{pc} \oplus \opc) \cup (\mathit{pc} \oplus [\delta, \delta])$ \quad otherwise \end{tabular} & Branch taken iff $\opa = \opb$ (unsigned). \\ \hline
\texttt{BNE} & \begin{tabular}[c]{@{}l@{}} $\mathit{pc} \oplus \opc$ \quad if $\widetilde{\mathsf{Eq}}(\opa, \opb) = [0,0]$;\\ $\mathit{pc} \oplus [\delta, \delta]$ \quad if $\widetilde{\mathsf{Eq}}(\opa, \opb) = [1,1]$;\\ $(\mathit{pc} \oplus \opc) \cup (\mathit{pc} \oplus [\delta, \delta])$ \quad otherwise \end{tabular} & Branch taken iff $\opa \ne \opb$. \\ \hline
\texttt{BGE} & \begin{tabular}[c]{@{}l@{}} $\mathit{pc} \oplus \opc$ \quad if $\widetilde{\mathsf{SLt}}(\opa, \opb) = [0,0]$;\\ $\mathit{pc} \oplus [\delta, \delta]$ \quad if $\widetilde{\mathsf{SLt}}(\opa, \opb) = [1,1]$;\\ $(\mathit{pc} \oplus \opc) \cup (\mathit{pc} \oplus [\delta, \delta])$ \quad otherwise \end{tabular} & Branch taken iff $\opa \ge_s \opb$ (signed). \\ \hline
\texttt{BGT} & \begin{tabular}[c]{@{}l@{}} $\mathit{pc} \oplus \opc$ \quad if $\widetilde{\mathsf{SLe}}(\opa, \opb) = [0,0]$;\\ $\mathit{pc} \oplus [\delta, \delta]$ \quad if $\widetilde{\mathsf{SLe}}(\opa, \opb) = [1,1]$;\\ $(\mathit{pc} \oplus \opc) \cup (\mathit{pc} \oplus [\delta, \delta])$ \quad otherwise \end{tabular} & Branch taken iff $\opa >_s \opb$ (signed). \\ \hline
\texttt{BLE} & \begin{tabular}[c]{@{}l@{}} $\mathit{pc} \oplus \opc$ \quad if $\widetilde{\mathsf{SLe}}(\opa, \opb) = [1,1]$;\\ $\mathit{pc} \oplus [\delta, \delta]$ \quad if $\widetilde{\mathsf{SLe}}(\opa, \opb) = [0,0]$;\\ $(\mathit{pc} \oplus \opc) \cup (\mathit{pc} \oplus [\delta, \delta])$ \quad otherwise \end{tabular} & Branch taken iff $\opa \le_s \opb$ (signed). \\ \hline
\texttt{BLT} & \begin{tabular}[c]{@{}l@{}} $\mathit{pc} \oplus \opc$ \quad if $\widetilde{\mathsf{SLt}}(\opa, \opb) = [1,1]$;\\ $\mathit{pc} \oplus [\delta, \delta]$ \quad if $\widetilde{\mathsf{SLt}}(\opa, \opb) = [0,0]$;\\ $(\mathit{pc} \oplus \opc) \cup (\mathit{pc} \oplus [\delta, \delta])$ \quad otherwise \end{tabular} & Branch taken iff $\opa <_s \opb$ (signed). \\ \hline
\texttt{BGEU} & \begin{tabular}[c]{@{}l@{}} $\mathit{pc} \oplus \opc$ \quad if $\widetilde{\mathsf{Lt}}(\opa, \opb) = [0,0]$;\\ $\mathit{pc} \oplus [\delta, \delta]$ \quad if $\widetilde{\mathsf{Lt}}(\opa, \opb) = [1,1]$;\\ $(\mathit{pc} \oplus \opc) \cup (\mathit{pc} \oplus [\delta, \delta])$ \quad otherwise \end{tabular} & Branch taken iff $\opa \ge_u \opb$ (unsigned). \\ \hline
\texttt{BLTU} & \begin{tabular}[c]{@{}l@{}} $\mathit{pc} \oplus \opc$ \quad if $\widetilde{\mathsf{Lt}}(\opa, \opb) = [1,1]$;\\ $\mathit{pc} \oplus [\delta, \delta]$ \quad if $\widetilde{\mathsf{Lt}}(\opa, \opb) = [0,0]$;\\ $(\mathit{pc} \oplus \opc) \cup (\mathit{pc} \oplus [\delta, \delta])$ \quad otherwise \end{tabular} & Branch taken iff $\opa <_u \opb$ (unsigned). \\ \hline
\texttt{JAL} & $\mathit{pc} \oplus \opb$ & Unconditional jump; target $= \mathit{pc} + \opb$ (signed offset). \\ \hline
\texttt{JALR} & $\widetilde{\mathsf{AddU}}(\opb, \opc)$ & Indirect jump; target $= (\opb + \opc) \bmod W$ (register + immediate). \\ \hline
\texttt{Jumpi} & $\opb$ & Absolute indirect jump; target is the value of register $\opb$ directly. \\
\bottomrule
\end{tabular}
\end{table*}

\paragraph{L-5) Memory Lookup}: The Memory Lookup constraint $\mathcal{L}_{\text{mem}}$ ensures global Read-After-Write (RAW) consistency. For each row $i$, let $\mathfrak{addr}_i(\tilde{T})$, $\mathfrak{val}_i(\tilde{T})$, and $\mathfrak{type}_i(\tilde{T})$ be the column projections for the address, value, and access type, respectively.

\sys maintains an Abstract Segmented Memory state $\hat{\mathcal{M}}$, which is updated or queried chronologically. The abstract evaluator $\widetilde{\mathcal{L}}_{\text{mem}}$ is defined by the aggregate result of a consistency check performed on each operation:

\begin{align}
&\widetilde{\mathcal{L}}_{\text{mem}}(\tilde{T}) = \\
&\begin{cases} 
\mathsf{True} & \text{if } \forall i, 1 \in \mathfrak{s}_i(\tilde{T}) \implies \\ &\textsf{IsRAM}(\mathfrak{addr}_i(\tilde{T}), \mathfrak{val}_i(\tilde{T}), \mathfrak{type}_i(\tilde{T}), \hat{\mathcal{M}}) = \mathsf{True} \\
\mathsf{False} & \text{if } \exists i, [1, 1] = \mathfrak{s}_i(\tilde{T}) \wedge \\ &\textsf{IsRAM}(\mathfrak{addr}_i(\tilde{T}), \mathfrak{val}_i(\tilde{T}), \mathfrak{type}_i(\tilde{T}), \hat{\mathcal{M}}) = \mathsf{False} \\
\mathsf{Uncertain} & \text{otherwise}
\end{cases}
\end{align}

\sys initializes $\hat{\mathcal{M}}$ with a single segment
$([0, 2^W - 1], [0, 0])$ representing zero-initialization of the
entire address space, so every address has a defined stored value
throughout the execution.

The internal logic of $\textsf{IsRAM}(\widehat{a}, \widehat{v}, \widehat{t}, \hat{\mathcal{M}})$ handles the stateful nature of
memory. For the write case, \textsf{IsRAM} updates $\hat{\mathcal{M}}$ by splitting any segments that intersect with $\widehat{a}$ and inserting the new value $\widehat{v}$. Since any write is locally valid, this operation always returns $\mathsf{True}$. For the read case, the evaluator verifies $\widehat{v}$ against the segments in $\hat{\mathcal{M}}$ that intersect with $\widehat{a}$:
\begin{itemize}
    \item Returns $\mathsf{True}$ if $\widehat{v}$ contains the
    stored value interval of every intersecting segment.
    \item Returns $\mathsf{False}$ if $\widehat{v}$ is disjoint from
    the stored value of some intersecting segment.
    \item Returns $\mathsf{Uncertain}$ otherwise.
\end{itemize}

\subsection{Proof of Theorem~\ref{thm:sc}}
\label{subsec:proof}

\begin{proof}[Proof Sketch]
The proof proceeds by establishing three properties: termination, soundness, and completeness.

\paragraph{Termination} The search space is a finite integer lattice $\mathbb{I}^{H \times W}$ defined over the range $[0, p-1]$. In each iteration of the branch-and-bound search, SplitInterval partitions a "divisible" cell ($a \neq b$) into strictly smaller sub-intervals. Because the volume of the intervals decreases monotonically and the search space is discrete and finite, the algorithm must terminate.

\paragraph{Local Soundness} By its definition, local soundness requires that every reported trace satisfies the constraints. In Algorithm 1, a trace $T$ is added to $Q$ only if $\tilde{C}(\tilde{T}) = \text{True}$. According to the Soundness Condition (Definition~\ref{def-soa}), if $\tilde{C}(\tilde{T}) = \text{True}$, then for all $T \in [[\tilde{T}]]_p$, it is guaranteed that $C(T) = \text{True}$. Thus, every element in $Q$ is a mathematically valid execution trace.

\paragraph{Local Completeness} Local Completeness requires that no valid trace is omitted. This is guaranteed by two mechanisms. 

First, the $\mathsf{SplitInterval}$ function maintains a total cover of the search space, ensuring the union of child nodes equals the parent node.

Second, A branch is only discarded if $\tilde{C}(\tilde{T}) = \text{False}$. By Definition~\ref{def-soa}, if $\tilde{C}(\tilde{T}) = \text{False}$, then no $T \in [[\tilde{T}]]_p$ exists such that $C(T) = \text{True}$. Therefore, pruning never discards a satisfying solution.

Finally, the canonicalizer $\mathcal{R}$ ensures that representational redundancies (such as padding or permutations) are mapped into a unique semantic normal form. This ensures that the cardinality of $Q$ accurately reflects the number of semantically distinct solutions, identifying under-constrained ($|Q| > 1$) or over-constrained ($|Q| = 0$) behaviors.
\end{proof}

\section{Implementation Details}
\label{appendix:implementation}

\subsection{Canonicalization}
\label{app:canonicalization}

{To ensure that our solution-set cardinality criterion $(\lvert Q \rvert)$ accurately distinguishes semantic bugs from representational artifacts, we formalize our canonicalizer in Lean 4. This lets us prove that the canonicalizer is invariant under representational artifacts of the table generator and one-to-one on semantically distinct executions.}

\smallskip
\noindent \textbf{Core Definitions.} {We define a \lstinline{Config} to encapsulate the domain-specific logic for identifying valid state transitions and projecting them into a normalized form. The canonicalizer $\mathcal{R}$ maps a trace to a unique set of semantic execution records.}

\begin{lstlisting}
structure Config (Repr : Type) where
  isReal : Row → Bool
  projectRow : Row → Repr

def canonicalize
    (cfg : Config Repr) (t : Trace) : Finset Repr :=
  ((t.filter cfg.isReal).map cfg.projectRow).toFinset
\end{lstlisting}

\noindent {The tuple type \lstinline{Repr} is per-table, for example, the ALU tuple (the first input, the second input, the output) for arithmetic instructions, or a control-flow tuple capturing program-counter transitions for branches and jumps.}

{Independent of the canonicalizer, we model the zkVM's notion of a semantic record for each table. \lstinline{Record} is the zkVM-level abstract semantics of an opcode execution, and the structure \lstinline{RecordIdentity Record Identity} carries an injection \lstinline{toIdentity} that assigns a unique identity to each record, witnessing that the zkVM's record type is fine-grained enough to distinguish every distinct execution of the opcode.}

\begin{lstlisting}
structure RecordIdentity (Record Identity : Type) where
  toIdentity : Record → Identity
  injective : Function.Injective toIdentity
\end{lstlisting}

\smallskip
\noindent \textbf{Faithful Table Generation Assumption}. {The single per-zkVM modeling assumption is that the table generator is \textit{faithful}: its non-padding projected rows reproduce exactly the records of the underlying execution, independently of row order, padding placement, or helper-column values.}

\begin{lstlisting}
def TableEncodesRecords
    (cfg : Config Identity) (recordId : RecordIdentity Record Identity)
    (records : RecordSet Record) (table : Trace) : Prop :=
  ∀ identity, identity ∈ Finset.image recordId.toIdentity records ↔
    ∃ row ∈ table, cfg.isReal row ∧ cfg.projectRow row = identity

def TableGeneratorFaithful
    (cfg : Config Identity) (recordId : RecordIdentity Record Identity)
    (generateTable : RecordSet Record → Trace) : Prop :=
  ∀ records, TableEncodesRecords cfg recordId records (generateTable records)
\end{lstlisting}

\smallskip
\noindent \textbf{Property 1: Generator-independence}. {We first prove that any two faithful generators, regardless of how they order rows or insert padding, result in the same canonicalized representation.}

\begin{lstlisting}
theorem canonicalize_generator_independent
    (cfg : Config Identity) (recordId : RecordIdentity Record Identity)
    (gen₁ gen₂ : RecordSet Record → Trace)
    (h₁ : TableGeneratorFaithful cfg recordId gen₁)
    (h₂ : TableGeneratorFaithful cfg recordId gen₂)
    (records : RecordSet Record) :
    canonicalize cfg (gen₁ records) = canonicalize cfg (gen₂ records) := by
  rw [canonicalize_generated_eq_recordIds cfg recordId gen₁ h₁ records,
      canonicalize_generated_eq_recordIds cfg recordId gen₂ h₂ records]
\end{lstlisting}

%\noindent This is the formal counterpart of \S3.4's collapse property: any structural difference between two faithful generators, such as padding, row order, helper-column population, is absorbed by the canonicalizer, so different raw representations of the same execution are mapped to a unique normal form.

\smallskip
\noindent\textbf{Property 2: Quotient--image bijection}. The canonical trace space $\widehat{\mathcal{D}} =
\mathcal{D}/{\sim_\mathcal{R}}$ of Definition~\ref{def:canonical-trace-table-space}, and the image $\text{Im}(\mathcal{R})$ are formalized as follows: % is in constructive bijection with the canonicalizer image $\mathrm{Im}(\mathcal{R})$.

\begin{lstlisting}
abbrev CanonicalTraceSpace (cfg : Config Repr) := 
    Quotient (semanticSetoid cfg)

abbrev CanonicalizerImage (cfg : Config Repr) := 
  {reprs : Finset Repr // ∃ table, 
   canonicalize cfg table = reprs }
\end{lstlisting}

\noindent {The map \lstinline{quotientToImage} then lifts \lstinline{canonicalize} through the quotient, sending each equivalence class to the canonical form shared by its representatives, and is proved to be bijective:}

\begin{lstlisting}
theorem canonicalTraceSpace_equiv_image 
  (cfg : Config Repr) :
  Function.Bijective (quotientToImage cfg)
\end{lstlisting}

\smallskip
\noindent \textbf{Property 3: Record-set injectivity.} {We further prove that distinct executions produce distinct canonical forms, provided the record encoding is injective.}

\begin{lstlisting}
theorem canonicalize_generated_eq_iff_records_eq
    (cfg : Config Identity) (recordId : RecordIdentity Record Identity)
    (generateTable : RecordSet Record → Trace)
    (hgen : TableGeneratorFaithful cfg recordId generateTable)
    (records₁ records₂ : RecordSet Record) :
    canonicalize cfg (generateTable records₁) =
      canonicalize cfg (generateTable records₂) ↔
    records₁ = records₂
\end{lstlisting}

\noindent {This theorem discharges the cardinality interpretation: $|Q|$ counts semantically distinct executions rather than representational variants. Injectivity of \lstinline|rid| is required because two different executions could otherwise collide on the same operand tuple.}

\smallskip
\noindent \textbf{Instantiation.} {Each evaluated zkVM specializes the framework by fixing the column layout of \lstinline|cfg| for each supported table. As an example, consider Fig.~\ref{fig:canonicalizer}, where the column map defines the purpose of each column, while the canonicalizer is shown in Rust-style pseudo-code. This figure is based on the actual implementation for SP1's AddSub table. This table identifies the first input operand as the
byte-limb group at columns~7--10, the second at columns~11--14, and the result at columns~0--3. The corresponding configuration assembles
these indices into the ALU operand tuple $(b, c, a)$:}

\begin{lstlisting}
def addConfig (isReal : Row → Bool) : Config ALU.Tuple :=
  mkALUConfig [7, 8, 9, 10] [11, 12, 13, 14] [0, 1, 2, 3] isReal
\end{lstlisting}

\noindent Then, the  record-set injectivity theorem states:

\begin{lstlisting}
theorem add_one_to_one (isReal : Row → Bool)
    (tableId : ALU.Tuple → Identity)
    (recordId : Generic.RecordIdentity Record Identity)
    (generateTable : Generic.RecordSet Record → Trace)
    (hgen : Generic.TableGeneratorFaithfulToIds
      (addConfig isReal) tableId recordId generateTable)
    (records₁ records₂ : Generic.RecordSet Record) :
    Generic.recordIdsOfTable (addConfig isReal) tableId (generateTable records₁) =
      Generic.recordIdsOfTable (addConfig isReal) tableId (generateTable records₂) ↔
    records₁ = records₂
\end{lstlisting}

\noindent {The corollary states that this add canonicalizer is one-to-one on record sets, conditional on the generator being faithful and on the record encoding being injective. Analogous corollaries hold for other zkVMs' instruction tables.}

\begin{figure}[!th]
    \centering
    \includegraphics[width=0.95\linewidth]{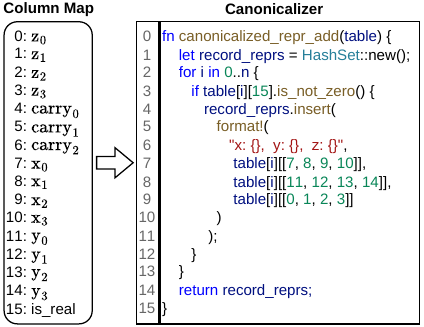}
    \caption{Pseudo-code of a canonicalizer ($\mathcal{R}$) for the ADD table. $\mathcal{R}$ projects raw trace tables into unique sets of semantic execution records, filtering out internal helper witnesses (e.g., carry bits) and padding rows (via the $\texttt{is\_real}$ selector). }
    \label{fig:canonicalizer}
\end{figure}

%This ensures a 1-to-1 mapping between the canonical trace space and the semantic execution history.

\subsection{Search Optimization}

\begin{lemma}[Soundness of \textsc{Conditional-Const-Gate}]
\label{lem:br-const-gate}
Let $\tilde{T}$ be an interval state, and let $s, v$ be cells of $\tilde{T}$
with $I_s = [1, 1]$. Suppose the constraint $s \cdot (v - \tau) = 0$ holds
over $\mathbb{F}_p$ for some constant $\tau \in \mathbb{F}_p$. Then the
refinement $I_v \leftarrow I_v \cap [\tau, \tau]$ is sound: for every
concrete $T \in \llbracket \tilde{T} \rrbracket_p$ satisfying the constraint,
$T$ remains in $\llbracket \tilde{T}' \rrbracket_p$, where $\tilde{T}'$ is
$\tilde{T}$ with $I_v$ replaced by the refined interval.
\end{lemma}

\begin{proof}
Fix any $T \in \llbracket \tilde{T} \rrbracket_p$ satisfying
$s \cdot (v - \tau) = 0$. Since $I_s = [1, 1]$, the concrete value of
$s$ in $T$ equals $1$, so the constraint reduces to $v - \tau = 0$ in
$\mathbb{F}_p$, i.e., $v = \tau$. Hence $v \in \llbracket [\tau, \tau]
\rrbracket_p$, and therefore $v \in \llbracket I_v \cap [\tau, \tau]
\rrbracket_p$. All other cells are unchanged, so
$T \in \llbracket \tilde{T}' \rrbracket_p$.
\end{proof}

\begin{lemma}[Soundness of \textsc{Conditional-Var-Gate}]
\label{lem:br-var-gate}
Let $\tilde{T}$ be an interval state, and let $s, a, b$ be cells of
$\tilde{T}$ with $I_s = [1, 1]$. Suppose the constraint
$s \cdot (a - b) = 0$ holds over $\mathbb{F}_p$. Then the refinement
$I_a, I_b \leftarrow I_a \cap I_b$ is sound.
\end{lemma}

\begin{proof}
Fix any $T \in \llbracket \tilde{T} \rrbracket_p$ satisfying
$s \cdot (a - b) = 0$. Since $I_s = [1, 1]$, we have $s = 1$ in $T$, and
the constraint reduces to $a = b$ in $\mathbb{F}_p$. Let $v$ denote this
common value. Then $v \in \llbracket I_a \rrbracket_p$ (since $a$ was
assigned from $I_a$) and $v \in \llbracket I_b \rrbracket_p$ (since
$b = a$ was assigned from $I_b$), so
$v \in \llbracket I_a \rrbracket_p \cap \llbracket I_b \rrbracket_p
\subseteq \llbracket I_a \cap I_b \rrbracket_p$. The same holds
symmetrically for $b$, so $T \in \llbracket \tilde{T}' \rrbracket_p$.
\end{proof}

\begin{lemma}[Soundness of \textsc{Affine-Backward}]
\label{lem:br-affine}
Let $\tilde{T}$ be an interval state, let $L$ be an affine expression
over cells of $\tilde{T}$, and let $a, e$ be cells with
$a, e \notin \mathrm{FV}(L)$. Let $d > 0$ be a constant, and let
$I_a = [a_{lo}, a_{hi}]$, $I_L = [\ell, h]$, $I_e = [e_{lo}, e_{hi}]$ be
the current intervals. Suppose the side condition
\[
|\ell| + |h| + d(|e_{lo}| + |e_{hi}|) + |a_{lo}| + |a_{hi}| < p
\tag{$\star$}
\]
holds, and the constraint $a = L - d \cdot e$ holds over $\mathbb{F}_p$.
Then the refinement
\[
I_e \;\leftarrow\; I_e \cap \left[
\left\lfloor \frac{\ell - a_{hi}}{d} \right\rfloor,\;
\left\lfloor \frac{h - a_{lo}}{d} \right\rfloor
\right]
\]
is sound.
\end{lemma}

\begin{proof}
Fix any $T \in \llbracket \tilde{T} \rrbracket_p$ satisfying
$a = L - d \cdot e$ in $\mathbb{F}_p$. By definition of
$\llbracket \cdot \rrbracket_p$, there exist integer witnesses
$\hat{a} \in [a_{lo}, a_{hi}]$ and $\hat{e} \in [e_{lo}, e_{hi}]$
reducing modulo $p$ to the values of $a$ and $e$ in $T$, and an
integer witness $\hat{L} \in [\ell, h]$ obtained by evaluating $L$
over the integer witnesses of its free variables (well-defined
since $a, e \notin \mathrm{FV}(L)$). By the side condition $(\star)$, the triangle inequality gives
\[
|\hat{L} - d \cdot \hat{e} - \hat{a}|
\leq |\hat{L}| + d \cdot |\hat{e}| + |\hat{a}|
\leq |\ell| + |h| + d(|e_{lo}| + |e_{hi}|) + |a_{lo}| + |a_{hi}| < p,
\]
so $\hat{L} - d \cdot \hat{e} - \hat{a} \in (-p, p)$. Since this integer
quantity reduces to $0$ modulo $p$ (by the constraint), it must equal
$0$ in $\mathbb{Z}$. Hence $\hat{a} = \hat{L} - d \cdot \hat{e}$ holds
over $\mathbb{Z}$, which rearranges to
$\hat{e} = (\hat{L} - \hat{a}) / d$.
Applying interval arithmetic over $\mathbb{Z}$, we have
$\hat{L} - \hat{a} \in [\ell - a_{hi},\; h - a_{lo}]$. Dividing by
$d > 0$ gives $\hat{e} \in [(\ell - a_{hi})/d,\; (h - a_{lo})/d]$. Since
$\hat{e}$ is an integer,
\[
\hat{e} \in
\left[\left\lfloor \frac{\ell - a_{hi}}{d} \right\rfloor,\;
\left\lfloor \frac{h - a_{lo}}{d} \right\rfloor\right].
\]
Combined with $\hat{e} \in \llbracket I_e \rrbracket_p$, we conclude
$\hat{e}$ lies in the refined interval, so
$T \in \llbracket \tilde{T}' \rrbracket_p$.
\end{proof}

\section{Experiment Details}
\label{appendix:experiemnt-details}

\paragraph{Configurations of SMT Baseline}

For the SMT baseline, each trace cell is materialized as a 32-bit bit-vector variable, and every AIR constraint is translated into its bit-vector equivalent, with field-prime equalities reduced via $\mathtt{bvurem}$.  An SMT solver, however, is an existential procedure: it returns a single satisfying assignment or certifies that none exists, but it cannot directly decide whether a constraint set is {tightly} determined by its inputs.  To turn it into an under-constraint detector, we first execute the zkVM's reference emulator to obtain an honest input--output pair $(\mathbf{x}^\star,\mathbf{y}^\star)$, then append a {blocking constraint} that excludes $\mathbf{y}^\star$ from the solver's output space.  The resulting query is satisfiable iff there exists an alternative trace that agrees with the honest input $\mathbf{x}^\star$ and the AIR but disagrees with $\mathbf{y}^\star$, i.e., a witness of under-constrained behavior; otherwise the query is unsatisfiable, certifying that no such bug exists for that input.  Each per-table formula is dispatched to Z3 under the same wall-clock timeout used for \sys.

\paragraph{Configurations of Fuzzer Baseline}

We use Arguzz, the state-of-the-art fuzzer for zkVMs~\cite{hochrainer2025arguzz}, as the common driver across all five targets. Arguzz combines metamorphic testing of zkVM circuit executions with controlled fault injection into the host-side executor, which lets it expose both completeness violations (the prover refuses an honestly executed program) and soundness violations (the verifier accepts a tampered execution).

SP1 and Pico are supported by the original Arguzz release, and we additionally implemented Arguzz front-ends for other Sphinx, Ziren, and Valida.   For each new target, we follow the extension architecture prescribed by Arguzz: a description of the target ISA's opcodes, an enumeration of the fault sites exposed by the host executor, a project generator that emits
a host/guest pair against the zkVM's SDK, and a small patch to the executor that realizes the fault primitives. For Ziren we replace Arguzz's RISC-V opcode table with the MIPS32r2 instruction set and retarget the syscall-identifier injection to MIPS' syscall convention.
For Valida whose ISA is a custom one, we also support basic transferable injection rules: program-counter modification, ALU result and operand modification, instruction replay, and syscall-identifier modification. In all cases we apply the safe remainder/division transformation provided by Arguzz to avoid spurious divide-by-zero crashes that would otherwise dominate the coverage signal.

All experiments are run with the same random seeds, on identical hardware to the one used for our \sys, and for the same wall-clock budget of six hours per target.

\paragraph{Input Range Verification via SMT solver}
\label{appendix:smt:range}

We also conduct the input range verification via the SMT solver. Specifically, the blocking-closure mechanism is repurposed accordingly: the verification starts from applying Z3 solver to the constraints without blocking closure, and each iteration extracts the satisfying input and output from Z3's returned model and appends a blocking constraint forbidding that input-output pair. The loop terminates once Z3 returns $\mathtt{unsat}$, where every input in $V$ admitting is expected to be enumerated. 

Tab.~\ref{tab:rq-appendix} reports the resulting speedups over the point-wise verification with the Z3 SMT solver, averaged over all opcodes per zkVM. Three observations stand out. First, the speedup never reaches $1\times$: across every zkVM and every $|V|$, the SMT range query is {slower} than the corresponding $|V|$ independent point queries. This is structural rather than incidental to our encoding. Blocking-clause enumeration must, by construction, issue at least $|V|$ \texttt{sat} calls (one per witness in $V$) plus a final \texttt{unsat} call to certify exhaustion, while accumulating $|V|$ blocking literals that progressively dilate each subsequent solve. Second, the trend is non-monotonic, peaking near $|V|\!=\!2^6$: for small $|V|$ the per-query process and parser overheads dominate the ratio, whereas for large $|V|$ the cumulative blocking-clause bloat slows late iterations enough to outweigh the amortization. Third, the largest volume $|V|\!=\!2^{14}$ exhausts the solver budget on every zkVM, and Valida fails to produce data for any $|V|$: sequential SMT enumeration does not scale to input regions of practical interest.

\begin{table}[!th]
\centering
\caption{
Speedup of Z3 SMT solver's \rangeverification over the iterative point-wise baseline, across input volumes $|V|$. T/O denotes the timeout. 
}
\label{tab:rq-appendix}
\begin{tabular}{c|cccc}
\toprule
    zkVM & $|V|=2^2$ & $|V|=2^6$ & $|V|=2^{10}$ & $|V|=2^{14}$ \\ \hline
Valida & T/O & T/O & T/O & T/O \\
Sphinx & 0.4$\times$ & 0.6$\times$ & 0.5$\times$ & T/O \\
Pico   & 0.6$\times$ & 0.6$\times$ & 0.5$\times$ & T/O \\
SP1    & 0.6$\times$ & 0.7$\times$ & 0.5$\times$ & T/O \\
Ziren  & 0.5$\times$ & 0.7$\times$ & 0.5$\times$ & T/O \\
 \hline
\end{tabular}
\end{table}
\else
  % Define appendix reference numbers without typesetting the appendix.
  \appendix

  % ------------------------------------------------------------
  % Appendix A: Overview of zkVMs
  % ------------------------------------------------------------
  \refstepcounter{section} % A
  \label{appendix:sp}

  \refstepcounter{table} % Table 6
  \label{tab:density}

  \refstepcounter{figure} % Figure 6
  \label{fig:examples-constraints}

  % ------------------------------------------------------------
  % Appendix B: Design Details
  % ------------------------------------------------------------
  \refstepcounter{section} % B

  \refstepcounter{figure} % Figure 7
  \label{fig:placeholder}

  \refstepcounter{subsection} % B.1
  \label{appendix:lookup}

  \refstepcounter{subsection} % B.2
  \label{appendix:subsec:ia}

  % Equations (11)--(17).
  % eq:abc refers to the final line of the byte-lookup definition,
  % which is Equation (17) in the current paper.
  \refstepcounter{equation} % 11
  \refstepcounter{equation} % 12
  \refstepcounter{equation} % 13
  \refstepcounter{equation} % 14
  \refstepcounter{equation} % 15
  \refstepcounter{equation} % 16
  \refstepcounter{equation} % 17
  \label{eq:abc}

  \refstepcounter{subsection} % B.3

  \refstepcounter{table} % Table 7
  \label{tab:star}

  \refstepcounter{table} % Table 8
  \label{tab:diamond}

  \refstepcounter{table} % Table 9
  \label{tab:cf}

  \refstepcounter{subsection} % B.4
  \label{subsec:proof}

  % ------------------------------------------------------------
  % Appendix C: Implementation Details
  % ------------------------------------------------------------
  \refstepcounter{section} % C
  \label{appendix:implementation}

  \refstepcounter{subsection} % C.1
  \label{app:canonicalization}

  \refstepcounter{figure} % Figure 8
  \label{fig:canonicalizer}

  \refstepcounter{subsection} % C.2

  \refstepcounter{theorem} % Lemma C.1
  \label{lem:br-const-gate}

  \refstepcounter{theorem} % Lemma C.2
  \label{lem:br-var-gate}

  \refstepcounter{theorem} % Lemma C.3
  \label{lem:br-affine}

  % ------------------------------------------------------------
  % Appendix D: Experiment Details
  % ------------------------------------------------------------
  \refstepcounter{section} % D
  \label{appendix:experiemnt-details}

  % This label occurs after an unnumbered paragraph, so it still
  % resolves to Appendix D.
  \label{appendix:smt:range}

  \refstepcounter{table} % Table 10
  \label{tab:rq-appendix}
\fi

\end{document}